\documentclass[a4paper,onecolumn,11pt,accepted=2022-06-28]{quantumarticle}
\pdfoutput=1
\usepackage[utf8]{inputenc}
\usepackage[english]{babel}
\usepackage[T1]{fontenc}
\usepackage{hyperref}
\usepackage{tikz}
\usepackage{lipsum}
\usepackage{graphicx}% Include figure files
\usepackage{dcolumn}% Align table columns on decimal point
\usepackage{bm}% bold math
\usepackage{braket}
\usepackage{amsmath}
\usepackage{bbold}
\usepackage{xcolor}
\usepackage{amsthm}
\usepackage{physics}
\newtheorem{theorem}{Theorem}
\newtheorem{proposition}{Proposition}
\newtheorem{corollary}{Corollary}[theorem]
\newtheorem{lemma}{Lemma}
\theoremstyle{definition}
\newtheorem{definition}{Definition}
\theoremstyle{remark}
\newtheorem*{remark}{Remark}
\newcommand{\Cov}{\mathrm{Cov}}
\newcommand{\Var}{\mathrm{Var}}

\begin{document}

%\title{Template demonstrating the quantumarticle document class}

\title{Work and Fluctuations: \newline Coherent vs. Incoherent Ergotropy Extraction}

\author{Marcin {\L}obejko}
\affiliation{Institute of Theoretical Physics and Astrophysics, Faculty of Mathematics, Physics and Informatics, University of Gda\'nsk, 80-308 Gdansk, Poland}
\affiliation{International Centre for Theory of Quantum Technologies, University of Gda\'nsk, 80-308 Gda\'nsk, Poland}
 \orcid{0000-0002-7159-5502}
\email{marcin.lobejko@ug.edu.pl}

\maketitle

\begin{abstract}
We consider a quasi-probability distribution of work for an isolated quantum system coupled to the energy-storage device given by the ideal weight. Specifically, we analyze a trade-off between changes in average energy and changes in weight's variance, where work is extracted from the coherent and incoherent ergotropy of the system. Primarily, we reveal that the extraction of positive coherent ergotropy can be accompanied by the reduction of work fluctuations (quantified by a variance loss) by utilizing the non-classical states of a work reservoir. On the other hand, we derive a fluctuation-decoherence relation for a quantum weight, defining a lower bound of its energy dispersion via a dumping function of the coherent contribution to the system's ergotropy. Specifically, it reveals that unlocking ergotropy from coherences results in high fluctuations, which diverge when the total coherent energy is unlocked. The proposed autonomous protocol of work extraction shows a significant difference between extracting coherent and incoherent ergotropy: The former can decrease the variance, but its absolute value diverges if more and more energy is extracted, whereas for the latter, the gain is always non-negative, but a total (incoherent) ergotropy can be extracted with finite work fluctuations. Furthermore, we present the framework in terms of the introduced quasi-probability distribution, which has a physical interpretation of its cumulants, is free from the invasive nature of measurements, and reduces to the two-point measurement scheme (TPM) for incoherent states. Finally, we analytically solve the work-variance trade-off for a qubit, explicitly revealing all the above quantum and classical regimes.  
\end{abstract}

\section{Introduction}
Miniaturization of energy-information processing devices is indisputably a cornerstone of nowadays technology. The growing ability to control and design systems ``at the bottom'' makes it possible to construct micro machines that transfer heat and work similarly to well-known macroscopic heat engines. However, thermodynamics at this scale faces two significant phenomena unobserved in the macroscopic reality. Firstly, the law of big numbers does not hold for small systems, which results in fluctuations of thermodynamic quantities comparable to the average values. Secondly, in the micro-world, quantum effects may play a significant role. In particular, the non-classical contributions from the quantum interference can appear in statistics of measured quantities.

Quantum thermodynamics is a program of adapting those two features into one universal framework.
The most significant milestone to incorporate the intrinsic randomness was the formulation of \emph{fluctuation theorems} \cite{Bochkov1977, Jarzynski1997, Crooks1999, Esposito2009, Campisi2011}, which are manifestations of the Laws of Thermodynamics expressed in terms of probability distributions. For quantum systems, the best-known concept of how to construct such distribution is the \emph{two-point measurement scheme} (TPM), leading to fluctuation theorems for work \cite{Talkner2007, Campisi2011, Esposito2009}, heat \cite{Jarzynski2004} or entropy production \cite{deffner2011, Manzano2018} (see also another proposals beyond the TPM, e.g., \cite{Ito2019, Debarba_2019, Micadei2020}). However, regarding the work distribution, quantum coherences in this approach are fundamentally rejected due to the invasive nature of the projective measurement. Hence, in order to incorporate quantum interference, other frameworks have been proposed like an idea of the \emph{work operator} \cite{Yukava2000, Allahverdyan2004work}. These attempts indeed can include quantum coherences; nevertheless, the obtained probability distributions are inconsistent with the fluctuation theorems if applied for incoherent states. Moreover, the work operator also has another drawback, i.e., it predicts zero work fluctuations for non-zero energy exchange for some initial states \cite{Allahverdyan2004work}. 

Thus, we observe some fundamental trade-offs in different definitions of work in a quantum domain. Later, this incompatibility has been rigorously formulated in the form of the no-go theorem \cite{Perarnau2017, Baumer2018}: Any generalized measurement of the work probability distribution that reproduces the average energy for all states (including those with coherences) is incompatible with fluctuation theorems (i.e., the TPM distribution) if applied to (initial) incoherent states. In the light of this result, another proposed solution, to quantitatively describe fluctuating quantities for coherent states, is to abandon the positivity of probability distributions and replace them with possibly negative-valued \emph{quasi-probabilities} \cite{Allahverdyan2014, Solinas2015, Solinas2016, Miller2017, Lostaglio2018, Levy2020}. In principle, these should recover the statistics enriched by interference terms, and for classical states, they should converge to distributions obeying classical fluctuations theorems (see the comprehensive review \cite{Baumer2018}).  

Despite the above obstacles, there is another fundamental problem in reconciling fluctuation theorems with quantum coherence. In all of the frameworks mentioned above, attempting to formulate the (quasi) distribution of work, the system on which the work is done (e.g., a load) is implicit and assumed, in general, to be an external classical field. As long as it is a good approximation for many cases, it is generally not adequate in plenty of experimental situations, where the energy-storage device should be treated autonomously within a fully quantum-mechanical framework. It turns out that treating the energy-storage device explicitly changes the physics of the work extraction process significantly. In particular, it leads to the so-called \emph{work-locking} \cite{Korzekwa2016, Lobejko2021} (i.e., an inability of work extraction from coherences), which non-autonomous frameworks cannot even capture. 

The problem of a proper model of the work reservoir is as hard as the problem of a correct work probability distribution. Nevertheless, we possess an essential hint since, in the first place, we need to make it consistent with (classical) fluctuation theorems. Here, the most crucial insight is realizing that the probability distribution constructed from the TPM solely relies on the energy differences rather than its absolute values \cite{Aberg2018, Alhambra2016}. For an explicit work reservoir, this is precisely equivalent to the translational invariance within the space of its energy states. The concept of such a device is called the \emph{quantum weight} and was first proposed in \cite{Brunner2012, Skrzypczyk2014}. Consequently, by this symmetry, we possess a fully quantum framework that recreates all of the fluctuation theorems (cf. \cite{Bartosik2021} where are discussed corrections if the symmetry is violated). Moreover, it was also shown that the optimal work extracted by the weight is limited by the \emph{ergotropy} \cite{Lobejko2020, Lobejko2021}, which, with the notion of passivity, is another building block of quantum thermodynamics (see, e.g., \cite{Pusz1978, Allahverdyan2004}). Hence, taking the weight as a proper model for fluctuation theorems allows us to analyze the quantum coherences from an autonomous perspective. So far, the work-locking for a quantum weight has been completely described as a decoherence process of a coherent contribution to the ergotropy, described by an effective, the so-called \emph{control-marginal state} \cite{Lobejko2021}. %Since this particular example of the work-locking refers to a loss of the ergotropy due to a dumping of coherences,  we call it the \emph{ergotropy decoherence}. 

In this paper, we take a step further, and instead of analyzing just the averages (i.e., the extracted work), we also formulate general results for the work fluctuations: the formula for changes of the weight's energy variance and bounds for its absolute values. Although we formulate our main results in a framework with an explicit work reservoir (i.e., the weight), firstly, we propose a general quasi-distribution constructed as a convolution of the one-point probability functions. The introduced quasi-distribution avoids the problem of the measurement invasiveness, and its statistical cumulants are related to changes of (one-point) cumulants (e.g., the first one is equal to work and the second to a change in the energy variance). Then, we adopt the concept to the quantum weight model, showing especially that in the classical (incoherent) limits, it converges to the TPM or work operator probability distribution. 

The main result of this research is the formula for change of the energy-storage variance (during the work extraction process). As it was shown in previous papers \cite{Lobejko2020, Lobejko2021}, the first cumulant of the quasi-distribution (i.e., change in the average energy of the weight) is equal to the first cumulant of the work operator but calculated for the control-marginal state. In the formula for the second cumulant (i.e., change in weight's variance), similarly appears an analogous term given by the second cumulant of the work operator; however, an additional term is also present, which is a fully quantum contribution that cannot be captured in non-autonomous frameworks. It is crucial since it may take negative values compared to the first one, and, as a consequence, it is possible to decrease the variance due to quantum interference. Moreover, the most surprising is that the second term can even overcome the first one, which results in a net decrease in energy dispersion. This entirely quantum feature, which leads to a new qualitative phenomenon (i.e., squeezing of work fluctuations), shows the importance of analyzing work distribution within a fully quantum framework. 

Secondly, we provide bounds for the weight's energy dispersion absolute value. In particular, we relate the bound for a final (or initial) dispersion with the dumping function that characterizes the aforementioned ergotropy decoherence process. Due to this interplay between work-locking caused by a decoherence process and dispersion of the energy, we call the proved inequality as the \emph{fluctuation-decoherence relation}. The main conclusion from this relation is that unlocking the system's total (coherent) ergotropy always results in a divergence of the weight's energy dispersion. Again, we stress that it is solely a quantum effect caused by the Heisenberg uncertainty principle. 

Finally, we apply the introduced general formulas to a qubit, and we derive an entire phase space of possible work-variance pairs for a quantum weight coupled to the arbitrary initial state of the system. Then, we go through a few particular examples and explicitly compare the classical and quantum regimes. In particular, we present a non-classical process of reducing the energy fluctuations for a `cat-state' of the weight (i.e., two energetically separate Gaussian wave packets), where both peaks collapse to each other (hence reducing the variance). In contrast, the classical protocol can only broaden each of them.

The paper is organized as follows. In Section \ref{work_definitions_section} we quickly review work distributions coming from the work operator and the TPM scheme, and then we propose a new quasi-distribution based on the convolution of one-point measurements. Section \ref{work_extraction_and_weight_section} is an introduction to the weight model, work extraction protocol, and all of the used mathematical methods. Next, in Section \ref{work_quasi_distribution_section} we apply a definition of the introduced quasi-distribution to a weight model and analyze its statistical properties and classical limits. In Section \ref{work_variance_gain_section} and \ref{bounds_section} we present our main results. First, we introduce the formulas for the average energy and variance changes and then the bounds for its absolute value (expressed as the fluctuation-decoherence relation). Finally, in Section \ref{qubit_section} we solve the problem of the work-variance trade-off for a qubit. We discuss quantum and classical regimes when coherent or incoherent ergotropy is extracted. We summarize the results with a discussion in Section \ref{summary_section}. 

\section{Work definitions} \label{work_definitions_section}

We start with a quick review of the work operator concept and the TPM probability distribution, which shall be used as our reference points. After that, in the subsection \ref{double_subsection}, we introduce a new definition of a quasi-probability based on the idea of the convolution of one-point distributions. 

\subsection{Preliminaries}
We shall start with a discussion of the non-autonomous closed system. In accordance, we assume that the initial state is described by the density matrix $\hat \rho_i$, and we consider a thermodynamic protocol given by the following unitary map, i.e.,
\begin{equation}\label{unitary_evolution}
\hat \rho_i \to \hat V \hat \rho_i \hat V^\dag.    
\end{equation} 
In general, the unitary $\hat V$ comes from an integration of the time-dependent Hamiltonian (describing the driving force), starting from the initial $\hat H_i = \sum_k \epsilon_k^i \Pi_k^i$, and resulting in the final operator $\hat H_f = \sum_k \epsilon_k^f \Pi_k^f$.

Throughout the paper we use a following definition of the $n$-th moment:
\begin{eqnarray}
\langle w^n \rangle_{\text{X}} = \int dw \ w^n \ P_X(w),
\end{eqnarray}
and the $n$-th cumulant:
\begin{equation} \label{cumulant_definition}
    \langle \langle w^n \rangle \rangle_{X} = \frac{d^n}{d(it)^n} \log  [\langle e^{i t w} \rangle_{X}] \Big{|}_{t=0},
\end{equation} 
for the probability distribution $P_X$. 

\subsection{Work operator}
Firstly, we review the idea of the work operator, which is defined as follows: 
\begin{equation} \label{averaged_work}
    \hat W = \hat H_i - \hat V^\dag \hat H_f \hat V.
\end{equation}
According to the introduced protocol \eqref{unitary_evolution}, since the system is isolated from the thermal environment, the mean change of its energy is identified with the average work, namely %$\langle w \rangle$, namely 
\begin{equation}
    \langle w \rangle_{\text{W}} = \Tr[\hat W \hat \rho_i].
\end{equation}
The important point is that the above definition includes the quantum contributions coming from the initial coherences  in a state $\hat \rho_i$. Furthermore, via the spectral decomposition of the work operator, i.e.,
\begin{equation}
    \hat W = \sum_i w_i \dyad{w_i},
\end{equation}
one can define the following probability distribution: 
\begin{equation} \label{work_operator_dist}
    P_{\text{W}}(w) = \sum_i \delta(w - w_i) \ \bra{w_i} \hat \rho_i \ket{w_i}. %, \ \ P_i = \Tr[\hat \rho_i \dyad{w_i}]
\end{equation}
The formula suggests that eigenvalues $w_i$ with associated probabilities $\bra{w_i} \hat \rho_i \ket{w_i}$ can be interpreted as the outcomes of the fluctuating work; however, according to the no-go theorem \cite{Perarnau2015} they do not satisfy the fluctuation theorems for classical systems (see e.g. \cite{Allahverdyan2014}).  

\subsection{(single) Two-point measurement (TPM)}
In order to construct the probability distribution satisfying fluctuation theorems, the TPM scheme was introduced \cite{Talkner2007, Campisi2011, Esposito2009}. In this approach, we perform the measurement on the initial state $\hat \rho_i$ (in its initial energy basis), which results with the outcome $\epsilon_n^i$, observed with the probability $p_n = \Tr[\hat \Pi_n^i \hat \rho_i]$. After that, the system is projected into the state: $\hat \rho_i \to \hat \Pi_n^i$. Next, we apply the protocol, i.e., $\Pi_n^i \to \hat V \hat \Pi_n^i \hat V^\dag$, and then the second measurement is performed, giving the outcome $\epsilon_m^f$ with the (conditional) probability $p_{m|n} = \Tr[\hat \Pi_m^f \hat V \hat \Pi_n^i \hat V^\dag]$. Finally, the fluctuating work is defined as the outcomes difference: $w = \epsilon_m^f - \epsilon_n^i$ with the associated (joint) probability $p_{m,n} = p_{m|n} p_n$.

According to the proposed two-point measurement scheme, the probability distribution of the fluctuating work is equal to: 
\begin{equation} \label{TPM_distribution}
    P_{\text{TPM}}(w) = \sum_{m,n} \delta(w - \epsilon_m^f + \epsilon_n^i) p_{m|n} p_n.
\end{equation}
The main problem with this approach is the invasive nature of the first measurement, such that all of the initial coherences are destroyed. As a consequence, in general the first moment (i.e., the average work) is incompatible with the value calculated via the work operator, namely 
\begin{equation}
    %\langle w \rangle_{\text{TPM}} = \int \ dw \ w \ P_{\text{TPM}}(w) \neq \Tr[\hat W \hat \rho_i] = \langle w \rangle_{\text{W}}.
    \langle w \rangle_{\text{TPM}}  \neq  \langle w \rangle_{\text{W}}.
\end{equation}

\subsection{(double) One-point measurement} \label{double_subsection}
To overcome the problem with the invasive nature of the first measurement, we introduce a quasi-probability distribution. Contrary to the TPM, here, the basic idea is to measure independently initial and final energy distributions and then construct a quasi-probability via the convolution operation.  

In accordance, firstly, we perform the initial measurement on the state  $\hat \rho_i$, which gives us outcomes $\epsilon_n^i$ with associated probabilities $p_n^i = \Tr[\hat \Pi_n^i \hat \rho_i]$. It is precisely the same procedure as for the TPM; however, in this case, instead of evolving the resulting (projected) state, it is evolved the initial (unperturbed) state $\hat \rho_i$. This results in a new statistical ensemble described by a density matrix $\hat \rho_f = \hat V \hat \rho_i \hat V^\dag$, and once again the statistics of final energy outcomes $\epsilon_n^f$ with probabilities $p_n^f = \Tr[\hat \Pi_n^f \hat V \hat \rho_i \hat V^\dag]$ is measured. Finally, we construct the initial and final energy distributions, i.e.,
\begin{equation}
    P_{i,f} (\epsilon) = \sum_n \delta(\epsilon - \epsilon_n^{i,f}) p_n^{i,f},
\end{equation}
and we introduce a real function $P_{\text{QP}}$, such that
\begin{equation}
    P_i (\epsilon_i) = \int d \epsilon_f P_{\text{QP}}(\epsilon_i - \epsilon_f) P_f (\epsilon_f).
\end{equation}
Finally, we interpret the kernel $P_{\text{QP}}(w)$ as the quasi-distribution of fluctuating work $w$, which formally is given by: 
\begin{equation} \label{quasi_probability}
    P_{\text{QP}} (w) = \frac{1}{2\pi} \int \ dt \  e^{-i w t} \frac{\Tr[e^{i \hat H_i t} \hat \rho_i]}{\Tr[e^{i \hat H_f t} \hat \rho_f]},
\end{equation}
where 
\begin{eqnarray}
\Tr[e^{i \hat H_{i,f} t} \hat \rho_{i,f}] = \int d \epsilon \ e^{i \epsilon t} P_{i,f} (\epsilon).
\end{eqnarray}
As was mentioned, the function $P_{\text{QP}}$ can have negative values, such that it cannot be interpreted as the proper probability distribution. Nevertheless, its $n$-th cumulant, i.e., $\langle \langle w^n \rangle \rangle_{\text{QP}}$ has a very straightforward interpretation as a difference of cumulants of the initial $P_i(\epsilon)$ and final $P_f(\epsilon)$ energy distributions, namely
\begin{equation}
\begin{split}
    \log[\langle e^{i t w} \rangle_{\text{QP}}] &= \log[\frac{1}{2\pi} \int ds \int d\omega \  e^{-i w (s-t)} \frac{\Tr[e^{i \hat H_i s} \hat \rho_i]}{\Tr[e^{i \hat H_f s} \hat \rho_f]}] \\
    &=\log[\Tr[e^{i \hat H_i t}]] - \log[\Tr[e^{i \hat H_f t}]]
\end{split}
\end{equation}
such that, according to the definition \eqref{cumulant_definition},
\begin{equation} \label{cumulant_difference}
    \langle \langle w^n \rangle \rangle_{\text{QP}} = \langle \langle \epsilon^n \rangle \rangle_i - \langle \langle \epsilon^n \rangle \rangle_f.
\end{equation}
In particular, for the first order ($n=1$), which is equal to the first moment, we have:
\begin{equation}
    %\langle w \rangle_{\text{QP}}  = \Tr[\hat H_i \hat \rho_i - \hat H_f \hat \rho_f] = \langle w \rangle_{\text{W}},
    \langle w \rangle_{\text{QP}}  = \langle w \rangle_{\text{W}},
\end{equation}
i.e., we recover the averaged work value given by the work operator \eqref{averaged_work}. 

The above results show that the quasi-distribution $P_{\text{QP}}$ encodes the proper work statistic enriched with quantum interference effects. Despite that we consider the work distribution here, we stress that definition \eqref{quasi_probability} can be used in general to quantify other thermodynamic quantities (like the heat flow). In the following sections, we will apply the introduced idea within the framework with an explicit work reservoir given by the quantum weight. In particular, we will prove that it reduces to the TPM scheme for incoherent states.

\section{Work extraction and quantum weight} \label{work_extraction_and_weight_section}
This section introduces a work extraction protocol for which we later define a quasi-probability distribution. In particular, the model is based on the First Law of Thermodynamics, i.e., strict energy conservation between coupled subsystems, and an explicit model of the energy-storage device known as the quantum weight \cite{Skrzypczyk2014, Alhambra2016, Lobejko2021}. 

We consider a composite Hilbert space $\mathcal{H}_{\mathcal{S}} \otimes \mathcal{H}_{\mathcal{W}}$ of the system $\mathcal{S}$ and the quantum weight $\mathcal{W}$. The former (system $\mathcal{S}$) is assumed to be finite-dimensional with a discrete energy spectrum, whereas the quantum weight $\mathcal{W}$ has a continuous and unbounded spectrum. Namely, the free Hamiltonians are given by:
\begin{equation}
    \hat H_S = \sum_i \epsilon_i \dyad{\epsilon_i}_S, \ \hat  H_W = \int d E \ E \dyad{E}_W.
\end{equation}

The work extraction protocol is understood as the energy transfer between subsystems, where a gain of the weight's average energy is interpreted as positive work. Throughout the paper, we assume an initial state given by a product $\hat \rho = \hat \rho_S \otimes \hat \rho_W$, and the unitary evolution, i.e., 
\begin{eqnarray}
\hat \rho_S \otimes \hat \rho_W \to \hat U \hat \rho_S \otimes \hat \rho_W \hat U^\dag,
\end{eqnarray}
where the operator $\hat U$  is specified in the next subsection. 

Throughout the paper, we denote by $S$ or $W$  subscript the operators that act entirely on the system or the weight Hilbert space, respectively, whereas operators without subscript act on both.

\subsection{Energy-conserving and translationally-invariant unitaries}
The crucial point of the whole model is to define a particular class of unitary operators such that the quantum weight can be interpreted as the work reservoir (i.e., the energy flowing from the system to the weight has no ``heat-like'' contribution). The weight model (defined below) overcame this work definition problem in the quantum regime. Essentially, firstly it was proven that it cannot decrease the entropy of the system \cite{Skrzypczyk2014, Alhambra2016}, and later, more precisely, that the optimal work is given by the system's ergotropy \cite{Lobejko2020, Lobejko2021}. 

The class of unitaries is defined by two symmetries: (i) the energy-conservation and (ii) the translational-invariance, such that the following commutation relations are satisfied: 
\begin{equation} \label{commutation_relations}
     [\hat U, \hat H_S + \hat  H_W ] = 0, \ [\hat U, \hat \Delta_W ] = 0,
\end{equation}
where $\hat \Delta_W$ is a generator of the displacement operator $\hat \Gamma_W(\epsilon) = e^{-i \hat \Delta_W \epsilon}$, obeying the canonical commutation relation $[ \hat H_W, \hat \Delta_W] = i$, such that $\hat \Gamma_W(\epsilon)^\dag \hat  H_W \hat \Gamma_W(\epsilon) = \hat  H_W + \epsilon$. We call the conjugate observable $\hat \Delta_W$ the time operator, and its eigenstates $\ket{t}_W$ (i.e., $\hat \Delta_W \ket{t}_W = t \ \ket{t}_W$), the time states. Then, we introduce the following probability density functions:
\begin{equation} \label{energy_time_dist}
%    f(E) = \Tr[\hat \rho_W \dyad{E}_W], \ \ g(t) = \Tr[\hat \rho_W \dyad{t}_W],
    f(E) = \bra{E} \hat \rho_W \ket{E}, \ \ g(t) = \bra{t} \hat \rho_W \ket{t},
\end{equation}
with the variances:
\begin{equation} \label{variance_time_energy}
    \begin{split}
        \sigma_E^2 &= \int dx \ x^2 \ f(x) - \left(\int dx \ x \ f(x) \right)^2, \\
        \sigma_t^2 &= \int dx \ x^2 \ g(x) - \left( \int dx \ x \ g(x) \right)^2.
    \end{split}
\end{equation}

Finally, any unitary operator $\hat U$ obeying commutation relations \eqref{commutation_relations} is given in the form  \cite{Alhambra2016, Lobejko2020}: 
\begin{equation} \label{unitary_with_S}
    \hat U = \hat S^\dag \hat V_S \hat S,
\end{equation}
with 
\begin{equation} \label{S_operator}
    \hat S = e^{-i \hat H_S \otimes \hat \Delta_W},
\end{equation}
and $\hat V_S$ is the unitary operator acting solely on the system Hilbert space $\mathcal{H}_\mathcal{S}$. 

\subsection{Coherent and incoherent work extraction}
In the following sections, we will see that the unitary $\hat V_S$ can be associated with the evolution operator introduced in Section \ref{work_definitions_section} for the non-autonomous framework with an implicit work reservoir. In analogy, we define the following work operator:
\begin{equation} \label{work_operator}
    \hat W_S = \hat H_S - \hat V_S^\dag \hat H_S \hat V_S.
\end{equation}
However, in contrast to Eq. \eqref{averaged_work}, here, the protocol is cyclic, such that the initial and final Hamiltonian is the same (i.e., given by $\hat H_S$). Consequently, the maximum change of the average energy with respect to all unitaries $\hat V_S$ is given by the ergotropy of the state $\hat \rho_S$ \cite{Allahverdyan2004}:
\begin{eqnarray} \label{ergotropy_definition}
R(\hat \rho_S) = \max_{\hat V_S} \Tr[\hat W_S \hat \rho_S].
\end{eqnarray}
Compared to Eq. \eqref{averaged_work}, one can get an impression that the autonomous framework presented here is less general since it is constrained only to the cyclic protocols. However, it can be easily generalized by adding the additional subsystem (i.e., the so-called `clock'), which controls changes of the system's Hamiltonian (see, e.g., \cite{Horodecki2013, Alhambra2016, Aberg2018}), or, another way around, one can assume that the $\mathcal{S}$ is a composite system (involving the proper one and the clock). 

In this framework, $\hat V_S$ is an arbitrary unitary operator acting on the system's Hilbert space. However, one can consider a subset of the so-called incoherent unitaries, denoted by $\hat V_S^I$, which members correspond to operations that permute the energy states (up to irrelevant phase factors) and thus preserve coherences \cite{Streltsov2017}. Within the energetic context, the \emph{incoherent work operator}, i.e., $\hat W_S^I = \hat H_S - \hat{V}_S^I{}^\dag \hat H_S \hat V_S^I$, satisfies the following commutation relation:
\begin{equation} \label{incoherent_work}
    [\hat W_S^I, \hat H_S] = 0,
\end{equation}
such that the average work is extracted solely from the diagonal part. Conversely, any operator \eqref{work_operator} that does not commute with the Hamiltonian $\hat H_S$ is the \emph{coherent work operator}, i.e., it affects coherences through the process of work extraction. 

According to this, the ergotropy \eqref{ergotropy_definition} can be divided into incoherent and coherent contribution, i.e., $R(\hat \rho_S) = R_I(\hat \rho_S) + R_C(\hat \rho_S)$. The former, being a part of energy extracted solely from a diagonal, is defined as:
\begin{eqnarray}
R_I(\hat \rho_S) = \max_{\hat V_S^I} \Tr[\hat W_S^I \hat \rho_S],
\end{eqnarray}
and then the coherent contribution is introduced by the formula:
\begin{eqnarray}
R_C(\hat \rho_S) = R(\hat \rho_S) - R_I(\hat \rho_S).
\end{eqnarray}

\subsection{Wigner function and control-marginal state}
To get a better insight of later calculations, we describe the weight state in terms of the Wigner function:
\begin{equation} \label{wigner_function}
    W(E,t) = \frac{1}{2\pi} \int d\omega \ e^{i \omega t} \Tr[\hat \rho_W \dyad{E + \frac{\omega}{2}}{E - \frac{\omega}{2}}_W],
\end{equation}
such that the probability density functions for energy and time states \eqref{energy_time_dist} are given by the marginals:
\begin{equation}
    f(E) = \int dt \ W(E,t), \ \ g(t) = \int dE \ W(E,t).
\end{equation}

Next, the crucial part of the introduced here framework is a definition of the effective density operator, which we call the \emph{control state} and \emph{control-marginal state} (see \cite{Lobejko2020, Lobejko2021}).
\begin{definition}[Control and control-marginal state] \label{control__marginal_state_def}
We define the \emph{control state}, according to the unitary transformation $\hat S$ \eqref{S_operator}, as:
\begin{equation} \label{control_state}
    \hat \sigma = \hat S \hat \rho \hat S^\dag.
\end{equation}
In accordance, an effective state of the system is the \emph{control-marginal state}:
\begin{equation} \label{control__marginal_state}
    \hat \sigma_S = \Tr_W[\hat S \hat \rho \hat S^\dag].
\end{equation}
\end{definition}
As we will see later, those states incorporates the crucial information on the work extraction process.

Furthermore, for the product state $\hat \rho = \hat \rho_S \otimes \hat \rho_W$, the control-marginal state is given by the mixture of free-dynamics unitaries averaged over a distribution of the weight's time states, namely
\begin{equation} \label{control_marginal}
    \hat \sigma_S = \int dt \ g(t) \ e^{-i \hat H_S t} \hat \rho_S e^{i \hat H_S t}.
\end{equation}
Since the channel is given as a convex combination of unitaries, it corresponds to the pure decoherence. 

\section{Work quasi-distribution} \label{work_quasi_distribution_section}
\subsection{Quasi-probability density function}
We are ready to apply the definition \eqref{quasi_probability} to the composite system with an explicit work reservoir. Based on the unitary protocol $\hat \rho_S \otimes \hat \rho_W \to \hat U \hat \rho_S \otimes \hat \rho_W \hat U^\dag$, we calculate the initial and final expected value of the operator $e^{i \hat H_W t}$, such that for the initial state we have
\begin{eqnarray}
\Tr[e^{i \hat H_W t} \hat \rho_S \otimes \hat \rho_W] = \Tr[e^{i \hat H_W t} \hat \rho_W],
\end{eqnarray}
whereas for the final:
\begin{eqnarray}
\Tr[e^{i \hat H_W t} \hat U \hat \rho_S \otimes \hat \rho_W \hat U^\dag]. 
\end{eqnarray}
According to definition \eqref{quasi_probability}, the quasi-probability distribution is given by the Fourier transform of the ratio of those two, i.e., we formally obtain the expression:
\begin{equation} \label{weight_quasi_probability}
    P_{\text{QP}}(w) = \frac{1}{2\pi} \int ds \ e^{-i w s} \frac{\Tr[\hat U^\dag e^{i \hat H_W s} \hat U \hat \rho_S \otimes \hat \rho_W]} {\Tr[ e^{i \hat H_W s} \hat \rho_W]}.
\end{equation}
Notice that in contrast to Eq. \eqref{quasi_probability}, the definition for the explicit energy storage has switched nominator with the denominator, such that the positive work refers to the energy gain of the weight. Next, since we assume that $\hat U$ obeys the commutation relations \eqref{commutation_relations}, we put the expression \eqref{unitary_with_S} and we get: 
\begin{equation} \label{quasi_dist_weight}
    %\frac{\Tr[\hat U^\dag e^{i \hat H_W s} \hat U \hat \rho_S \otimes \hat \rho_W]}{\Tr[ e^{i \hat H_W s} \hat \rho_W]} 
    P_{\text{QP}}(w)  = \frac{1}{2\pi} \int ds \ e^{-i w s} \Tr[\hat{\mathcal{M}}_S(s) \hat \xi_S (s)],
\end{equation}
where 
\begin{equation}
    \hat{\mathcal{M}}_S(s) = e^{\frac{1}{2} i \hat H_S s} e^{-i \hat V_S^\dag \hat H_S \hat V_S s} e^{\frac{1}{2} i \hat H_S s},
\end{equation}
and
\begin{equation} \label{xi_state}
    \hat \xi_S (s) = \frac{\int dt \int dE \ e^{i E s} \ W(E,t) \ e^{-i \hat H_S t} \hat \rho_S e^{i \hat H_S t}}{\Tr[ e^{i \hat H_W s} \hat \rho_W]},
\end{equation}
where the Wigner function \eqref{wigner_function} was introduced.

\subsection{Semi-classical and incoherent states}
Let us consider a family of operators given by Eq. \eqref{xi_state} (for $s \in \mathbb{R}$). The structure of the trace average in Eq. \eqref{quasi_dist_weight} suggests that $\hat \xi(s)$ is an effective density matrix; however, even though $\Tr[\hat \xi_S (s)] = 1$, the operator $\hat \xi_S (s)$ is not necessarily positive semi-definite. 

Nevertheless, according to the operator $\hat \xi(s)$, we define the class of the \emph{semi-classical states}. 

\begin{definition}[Semi-classical states] \label{semi_clasical_definiton_box}
We characterize the initial density operator $\hat \rho_S \otimes \hat \rho_W$ as the semi-classical if the state $\hat \xi_S (s)$ does not depend on the variable $s$ (i.e., $\hat \xi_S'(s) = 0$). Moreover, since we have $\hat \xi(0) = \hat \sigma_S$, then the semi-classical state can be simply defined by the condition:
\begin{equation} \label{semi_classical_def}
    \hat \xi_S (s) = \hat \sigma_S,
\end{equation}
where $\hat \sigma_S$ is the introduced control-marginal state \eqref{control_state}.
\end{definition}

We call the class of states obeying Eq. \eqref{semi_classical_def} the semi-classicals because it includes the set of Gaussian wave-packets (with non-negative Wigner function) as well as the incoherent states. Later, we will see that those states also predict non-negative changes in the work variance, contrary to the non-classical states, for which the negative contribution appears due to quantum interference. One should notice that the ``semi-classical state'', in general, refers to the composite state $\hat \rho_S \otimes \hat \rho_W$. However, sometimes it is highlighted which system is responsible for obeying the condition \eqref{semi_classical_def}; hence, we refer to either ``system'' or ``weight'' semi-classical state, respectively.

\emph{Gaussian wave-packets.} Let us start with an important example of the weight semi-classical states $\hat \rho_W$. First, notice that if the Wigner function \eqref{wigner_function} of the weight factorizes into the product of the energy and time states distributions, i.e., 
\begin{equation} \label{product_E_t}
    W(E,t) = f(E) g(t),
\end{equation}
then 
\begin{eqnarray}
    \int dt \int dE \ e^{i E s} \ W(E,t) \ e^{-i \hat H_S t} \hat \rho_S e^{i \hat H_S t} &=& \int dE \ e^{i E s} f(E) \int dt \ g(t) \ e^{-i \hat H_S t} \hat \rho_S e^{i \hat H_S t} \nonumber \\
     &=& \Tr[ e^{i \hat H_W s} \hat \rho_W] \hat \sigma_S,
\end{eqnarray}
such that the condition \eqref{semi_classical_def} is satisfied if the expression is put into the definition \eqref{xi_state}. Moreover, according to Hudson's theorem for pure states \cite{Hudson1974}, the separability condition \eqref{product_E_t} is satisfied if and only if the wave function of the weight has the Gaussian form, namely
\begin{equation} \label{gaussian_states}
    %\psi_W (E) = e^{-a E^2 + b E + c}
    \psi_{\mu, \nu, \sigma} (E) = (2 \pi \sigma^2)^{-\frac{1}{4}} e^{-\frac{1}{4 \sigma^2} (E-\mu)^2 + i \nu E}
\end{equation}
with real $\mu$ and $\nu$. Thus, the Gaussian wave packets belong to the class of semi-classical states. Another example is the set of states with uniform wave functions.

\emph{Incoherent states of the system.} These states are simply defined by the following commutation relation:
\begin{equation} \label{incoherent_states}
    [\hat \rho_S, \hat H_S] = 0. %\ \text{or} \ [\hat \rho_W, \hat H_W] = 0.
\end{equation}
Since those commute also with the unitary $e^{-i \hat H_S t}$, we get
\begin{multline}
    \int dt \int dE \ e^{i E s} \ W(E,t) \ e^{-i \hat H_S t} \hat \rho_S e^{i \hat H_S t} = \int dE \ e^{i E s} \ f(E) \ \hat \rho_S = \Tr[ e^{i \hat H_W s} \hat \rho_W] \hat \rho_S,
\end{multline}
and according to Eq. \eqref{control_marginal} we have $\hat \sigma_S = \hat \rho_S$, i.e., an arbitrary incoherent state of the system $\hat \rho_S$ is a semi-classical in accordance to Eq. \eqref{semi_classical_def}.

\emph{Incoherent states of the weight.} The definition of the incoherent state of the weight is less evident since here we consider the system with a continuous energy spectrum, and therefore the incoherent state of the weight is non-normalizable. Nevertheless, one can still consider the limit of a Gaussian wave packet \eqref{gaussian_states} with vanishing variance, i.e., $\sigma \to 0$,
%\begin{eqnarray}
%\psi_{\text{inc}} = \lim_{\sigma \to 0} \psi_{\mu, \nu, \sigma},
%\end{eqnarray}}
which, as it is proved above, is the semi-classical state. Notice that for a Gaussian wave packet, the dispersion of time and energy states \eqref{variance_time_energy} obey the minimal uncertainty relation, i.e., $\sigma_t \sigma_E = \frac{1}{2}$, such that in the limit of vanishing dispersion of energies, the dispersion of time states diverges. As a consequence, the channel given by Eq. \eqref{control_marginal} is fully depolarizing, such that 
\begin{eqnarray} \label{weight_incoherent_dephasing}
\lim_{\sigma \to 0} \psi_{\mu, \nu, \sigma} \implies \hat \sigma_S \to D[\hat \rho_S],
\end{eqnarray} 
where $D[\cdot]$ is a dephasing in the energy basis of the system. The generalization to mixed (incoherent) states is straightforward by taking the initial convex combination of the Gaussian states. 

\subsection{Classical limits}

We are ready to analyse the work statistics encoded in the quasi-distribution $P_{\text{QP}}$ \eqref{quasi_dist_weight} in the classical limit of incoherent states or/and unitary evolutions. In particular, in the following proposition we relate the quasi-distribution $P_{\text{QP}}$ with the TPM $P_{\text{TPM}}$ \eqref{TPM_distribution} and the work-operator distribution $P_{\text{W}}$ \eqref{work_operator_dist}.

\begin{proposition} \label{TPM_limit}
Let us consider the following work distributions: 
\begin{enumerate}
    \item $P_{\text{QP}}$ for the unitary $\hat U = \hat S^\dag \hat V_S \hat S$ and state $\hat \rho_S \otimes \hat \rho_W$,
    \item $P_{\text{W}}$ for the work operator $\hat W_S = \hat H_S - \hat V_S^\dag \hat H_S \hat V_S$ and state $\hat \rho_S$,
    \item $P_{\text{TPM}}$ for the unitary $\hat V_S$ and state $\hat \rho_S$.
\end{enumerate}
Then,
\begin{equation}
    P_{\text{QP}} = P_{\text{TPM}}
\end{equation}
if the initial state is incoherent, and 
\begin{equation} \label{QP_TPM_W}
    P_{\text{QP}} = P_{\text{TPM}} = P_{\text{W}}
\end{equation}
if the work operator is incoherent. 
\end{proposition}

\begin{proof}
Firstly, let us remind that for an incoherent state of the system (i.e., obeying \eqref{incoherent_states}) we have $\hat \xi_S (s) = \hat \rho_S = D[\hat \rho_S]$. Similarly, in the limit of the incoherent weight state, we get $\hat \xi_S (s) = \hat \sigma_S \to D[\hat \rho_S]$ (see Eq. \eqref{weight_incoherent_dephasing}). Taking it into account, for an arbitrary incoherent state, the quasi-distribution is equal to:
\begin{equation}
\begin{split}
     P_{\text{QP}}(w) &= \frac{1}{2\pi} \int ds \ e^{-i w s} \Tr[\hat{\mathcal{M}}_S(s) D[\hat \rho_S]] \\
     &= \frac{1}{2\pi} \int ds \ e^{-i w s} \Tr[e^{i \hat H_S s} e^{-i \hat V_S^\dag \hat H_S \hat V_S s}  D[\hat \rho_S]] \\
     &= \frac{1}{2\pi} \int ds \sum_n e^{-i (w - \epsilon_n) s} \bra{\epsilon_n} \hat V_S^\dag e^{-i \hat H_S s} \hat V_S  \ket{\epsilon_n} p_n\\
     &= \frac{1}{2\pi} \int ds \sum_{n,m} e^{-i (w - \epsilon_n + \epsilon_m) s} |\bra{\epsilon_m} \hat V_S \ket{\epsilon_n}|^2 p_n\\
     &= \sum_{m,n} \delta (w + \epsilon_m - \epsilon_n) p_{m|n} p_n = P_{\text{TPM}} (w), 
\end{split}
\end{equation}
where $p_{m|n} = |\bra{\epsilon_m} \hat V_S \ket{\epsilon_n}|^2$ and $p_n = \bra{\epsilon_n} D[\hat \rho_S] \ket{\epsilon_n}$.

Next, let us consider the incoherent work operator with a spectral decomposition: $\hat W_S^I = \sum_n w_n \ \dyad{w_n}_S$. We observe that:
\begin{equation} \label{incoherent_exponent}
    \hat{\mathcal{M}}_S(s) = e^{i \hat W_S^I s}, 
\end{equation}
and then we have
\begin{equation}
\begin{split}
     P_{\text{QP}}(w) &= \frac{1}{2\pi} \int ds \ e^{-i w s} \Tr[e^{i \hat W_S^I s} \hat \xi_S(s)] \\
     &= \frac{1}{2\pi} \int ds \ e^{-i w s} \Tr[e^{i \hat W_S^I s} D[\hat \rho_S]] \\
     &= \sum_n \frac{1}{2\pi} \int ds \ e^{-i (w-w_n) s} \bra{w_n} D[\hat \rho_S] \ket{w_n} \\
     &= \sum_n \delta(w-w_n) P_n = P_{\text{W}}(w),
\end{split}
\end{equation}
where $P_n = \bra{w_n} D[\hat \rho_S] \ket{w_n}$.
\end{proof} 

We stress that convergence to the TPM for incoherent states is one of the requirements imposed on quasi-distributions (see \cite{Perarnau2015, Baumer2018}) since it recovers the classical fluctuations theorems. 

\subsection{Moments}
Secondly, we analyse moments of the quasi-distribution $P_{\text{QP}}$. Its characteristic function is given by:
\begin{equation}
\begin{split}
    \langle e^{i w t} \rangle_{\text{QP}} &= \int dw \ e^{iwt}  P_{\text{QP}}(w) = \Tr[\hat{\mathcal{M}}_S(t) \hat \xi_S (t)] \\
    %\hat{\mathcal{M}}(s) &= e^{\frac{1}{2} i \hat H_S s} e^{-i \hat V_S^\dag \hat H_S \hat V_S s} e^{\frac{1}{2} i \hat H_S s}.
\end{split}
\end{equation}
such that statistical moments can be calculated via the expression: 
\begin{equation} \label{moment_def}
    \langle w^n \rangle_{\text{QP}} = \frac{1}{i^n} \frac{d^n}{dt^n} \Tr[\hat{\mathcal{M}}_S(t) \hat \xi_S (t)] \Big{|}_{t=0}.
\end{equation}
In general, the derivative in the above expression is calculated for a product of the operator $\hat{\mathcal{M}}_S(t)$ and the effective operator $\hat \xi_S (t)$, which, as we will see later, has important implications for the work extraction from coherences. 

In the light of Proposition \ref{TPM_limit}, for incoherent states or work operators, we have 
\begin{equation}
    \langle e^{i w t} \rangle_{\text{QP}} = \Tr[\hat{\mathcal{M}}_S(t) D[\hat \rho_S]],
\end{equation}
such that the derivative in Eq. \eqref{moment_def} is only calculated with respect to the operator $\hat{\mathcal{M}}_S(t)$. As expected, the initial coherences (even if present) do not contribute. Additionally, if the work operator is incoherent (i.e., commutes with the Hamiltonian $\hat H_S$), we get the following simplified expression:
\begin{equation}
    \langle w^n \rangle_{\text{QP}} = \Tr[(\hat W_S^I)^n  \hat \rho_S],
\end{equation}
which, in this purely classical regime, relates the $n$-th moment of the quasi-distribution $P_{\text{QP}}$ to the $n$-th power of the incoherent work operator (averaged over the initial density matrix).

In general, for the coherent work operator (i.e., with the non-vanishing commutator \eqref{incoherent_work}), the expansion of the operator $\hat{\mathcal{M}}_S(t)$ is equal to:
\begin{equation} \label{moments_expansion}
\begin{split}
     \hat{\mathcal{M}}_S(t) &= \mathbb{1} + i t \hat W_S + \frac{(it)^2}{2} \hat W_S^2 +\frac{(it)^3}{3!} \hat W_S^3\\
     &+\frac{(it)^3}{3!} \frac{1}{2} \left[[\hat H_S, \hat V_S^\dag \hat H_S \hat V_S], \hat V_S^\dag \hat H_S \hat V_S \right] \\
     &+ \frac{(it)^3}{3!} \frac{1}{4} \left[\hat H_S, [\hat H_S, \hat V_S^\dag \hat H_S \hat V_S] \right] + \dots,
\end{split}
\end{equation}
which shows that the non-commuting contributions appears only in the third and higher orders. Consequently, we arrive with simple formulas for the first two moments, namely
\begin{equation} \label{moments}
\begin{split}
    \langle w \rangle_{\text{QP}} &= \Tr[\hat W_S \hat \xi_S (0)], \\
    \langle w^2 \rangle_{\text{QP}} &= \Tr[\hat W_S^2 \hat \xi_S (0)] - 2i \Tr[\hat W_S \hat \xi'_S (0)],
\end{split}
\end{equation}
where $\hat \xi_S (0) = \hat \sigma_S$ and $\xi'_S (0)$ is the derivative evaluated in the point $s=0$. 

As we will see in the next section, the second term appearing in the expression for the second moment has very interesting implications: its non-vanishing value is the signature of the non-classical state of the work reservoir.

\section{Work vs. variance gain} \label{work_variance_gain_section}
In the previous section, we have discussed moments of the distribution $P_{\text{QP}}$ and its relation to the work operator \eqref{work_operator}. However, based on a definition of the double one-point distribution \eqref{quasi_probability}, the cumulants are more interesting due to their physical interpretation \eqref{cumulant_difference}. For the explicit work reservoir, these are given by a difference between cumulants calculated for the initial and final state of the energy storage. We introduce the following notation:
\begin{equation}
\begin{split}
    \Cov_{\hat \rho}[\hat A, \hat B] &= \frac{1}{2} \langle \hat A \hat B + \hat B \hat A \rangle_{\hat \rho} - \langle \hat A \rangle_{\hat \rho} \langle \hat B \rangle_{\hat \rho}, \\
    \Var_{\hat \rho}[\hat A] &= \Cov[\hat A, \hat A]_{\hat \rho}, \ \ \langle \hat A \rangle_{\hat \rho} = \Tr[\hat A \hat \rho],
\end{split}
\end{equation}    
and concentrate on first two cumulants. 

The first is given as a change in average energy of the weight (i.e., the extracted work):
\begin{equation} \label{Ew}
\begin{split}
    %\Delta E_W  &:= \int dw \ w \ P_{\text{QP}} (w) = \langle \hat U^\dag \hat  H_W \hat U \rangle_{\hat \rho}  - \langle \hat  H_W \rangle_{\hat \rho},
    \langle\langle w \rangle\rangle_{\text{QP}} &= \langle w \rangle_{\text{QP}} = \langle \hat U^\dag \hat  H_W \hat U \rangle_{\hat \rho}  - \langle \hat  H_W \rangle_{\hat \rho} \equiv \Delta E_W,
\end{split}
\end{equation}
whereas the second is equal to a change of the weight variance: 
\begin{equation} \label{sigmaW}
\begin{split}
    %\Delta \sigma_W^2 &:= \int dw \ w^2 \ P_{\text{QP}} (w) - \left(\int dw \ w \ P_{\text{QP}} (w) \right)^2\\
    \langle\langle w^2 \rangle\rangle_{\text{QP}} &= \langle w^2 \rangle_{\text{QP}} - \langle w \rangle_{\text{QP}}^2 =\Var_{\hat \rho} [\hat U^\dag \hat  H_W \hat U] - \Var_{\hat \rho} [\hat  H_W ] \equiv \Delta \sigma_W^2.
\end{split}
\end{equation}

We are now ready to formulate our first main result, which characterizes changes in the weight's energy and variance through the work extraction protocol. In the following we put $\hat W_S(t) = e^{i \hat H_S t} \hat W_S e^{-i \hat H_S t}$. We refer the reader to previously introduced Definitions \ref{control__marginal_state_def} and \ref{semi_clasical_definiton_box}, that are crucial for further considerations.
\begin{theorem} \label{work_variance_theorem}
The change of the average energy $\Delta E_W$ and energy variance $\Delta \sigma_W^2$ of the quantum weight is equal to:
\begin{equation} \label{work_variance_equations}
\begin{split}
    \Delta E_W &=  \langle \hat W_S \rangle_{\hat \sigma_S}, \\
    \Delta \sigma_W^2 &= \Var_{\hat \sigma_S}[\hat W_S] + 2 F,
\end{split}
\end{equation}
where  
\begin{equation}
\begin{split}
    F &= \Cov_{\hat \sigma} [\hat H_S - \hat H_W, \hat V_S^\dag \hat H_S \hat V_S] = -i \Tr[\hat W_S \hat \xi'_S (0)] \\
    &= \int dt \Tr[\hat W_S(t) \hat \rho_S] \int dE \ [E - \langle \hat H_W \rangle_{\hat \rho}] \ W(E,t).
    %- \Cov_{\hat \sigma} [\hat V_S^\dag \hat H_S \hat V_S, \hat H_W]
    %G &= \Cov_{\hat \sigma} [\hat H_W, \hat W]
\end{split}
\end{equation}
%and $W_S(t) = e^{i \hat H_S t} \hat W_S e^{-i \hat H_S t}$.
\end{theorem}
Firstly, let us observe that the extracted work $\Delta E_W$ is equal to the average of the work operator $\hat W_S$ with respect to the control-marginal state $\hat \sigma_S$, such that the optimal value is precisely given by its ergotropy \eqref{ergotropy_definition}, i.e., $\Delta E_W \le R(\hat \sigma_S)$. In particular, the replacement of the proper marginal state by the control-marginal has huge implications on work extraction from coherences (see \cite{Lobejko2020, Lobejko2021}), which is discussed in more detail in Section \ref{ergotropy_decoherence}.  

Secondly, we see that there are two contributions to the change of the variance. The first one, similarly to the average energy, is given by the variance of the work operator and calculated with respect to the control-marginal state. One should notice that those two terms (i.e., $\langle \hat W_S \rangle_{\hat \sigma_S}$ and $\Var_{\hat \sigma_S}[\hat W_S]$) are solely calculated within the system's Hilbert space, and in this sense, they could be compared to the non-autonomous protocols like the TPM measurements. However, they still might be influenced by coherences (present in the control-marginal state), which are affected by the weight via a decoherence process \eqref{control_marginal}. Thus, they generally involve additional quantum effects and information from a work reservoir absent in non-autonomous frameworks.  

Finally, the last $F$-term, i.e., the second contribution to the variance gain, is the most complex, primarily because it is evaluated in the total Hilbert space (e.g., via the control state $\hat \sigma$ \eqref{control_state}). For this reason, in principle, it cannot be captured by the frameworks treating the energy-storage device implicitly, and as we will see in the next section, it has solely quantum origin. The $F$-term is presented in three different forms. The first relates it to the covariance difference between initial Hamiltonians (i.e., $\hat H_S$ and $\hat H_W$) and the final one $\hat V_S^\dag \hat H_S \hat V_S$ calculated for the control state $\hat \sigma$ \footnote{In fact, the whole right-hand side of Eq. \eqref{work_variance_equations} is evaluated with respect to $\hat \sigma$, but terms other than the $F$-term are calculated for local operators.}. The next formula involves a definition of $\hat \xi(s)$ operator given by Eq. \eqref{xi_state}, and it straightforwardly follows from the expression for the moments \eqref{moments} and definition of the second cumulant \eqref{sigmaW}. Finally, the third expression is derived from Wigner's representation of the weight state \eqref{wigner_function}. It is especially interesting since the parameter $t$ refers here to a `time' of a system's free evolution. The $F$-term is given by a product of an expected value of the work operator in a Heisenberg picture (i.e., $\text{Tr} [\hat W_S(t) \hat \rho_S]$) and the weight's energy deviations (i.e., $E - \langle \hat H_W \rangle_{\hat \rho_W}$), averaged over a Wigner quasi-distribution $W(E,t)$. 

\begin{remark}
Notice that Theorem \ref{work_variance_theorem} can be expressed in the following form:
\begin{eqnarray}
\langle \langle w  \rangle \rangle_{\text{QP}, \hat \rho} &=& \langle \langle w  \rangle \rangle_{\text{W}, \hat \sigma_S}, \\
\langle \langle w^2  \rangle \rangle_{\text{QP}, \hat \rho} &=& \langle \langle w^2  \rangle \rangle_{\text{W}, \hat \sigma_S} + 2F,
\end{eqnarray}
where on the left-hand side we have cumulants of quasi-distribution \eqref{quasi_dist_weight} (defined for a composite state $\hat \rho$), whereas on the right-hand side there are cumulants of the work operator distribution \eqref{work_operator_dist} (defined for a control-marginal state $\hat \sigma_S$). However, despite this nice correspondence, due to a presence of the $F$-term, it is apparent that probability density function of the work operator $P_W$ cannot alone properly described the energy fluctuations. Another subtlety is that cumulants on the right-hand side are calculated for a control-marginal state, which is different than the initial state $\hat \rho_S$ (i.e., it is affected by a decoherence process depending on the initial state of the weight).  
\end{remark}

\subsection{Variance changes}
Let us now discuss some of the consequences of Theorem \ref{work_variance_theorem}.
According to the definition \eqref{semi_classical_def}, for semi-classical states we have $\hat \xi_S' (0) = 0$, such that $F=0$, and we conclude:

\begin{corollary} \label{positive_variance}
For semi-classical states, the variance change is always non-negative and equal to:
\begin{equation}
    \Delta \sigma_W^2 = \Var_{\hat \sigma_S}[\hat W_S] \ge 0.
\end{equation}  
\end{corollary}
We interpret it as one of the main features of the semi-classical states. On the contrary, for the non-classicals we have:
\begin{corollary}
For non-classical states, the $F$-term can be negative and, in particular, can lead to the negative change in variance (i.e., $\Delta \sigma_W^2 < 0$). 
\end{corollary}
We will show that by analyzing a particular example discussed in Section \ref{qubit_section}. Due to this feature, we see that non-classical states are qualitatively different from semi-classical ones. We stress that the squeezing of energy fluctuations is a pure quantum effect that involves interference between the coherent state of the system and the weight. 

Going back to the semi-classical states, we further ask the conditions for a null change of the variance.
\begin{corollary} \label{corollary_zero_variance}
For semi-classical states, a change of the weight variance is zero (i.e., $\Delta \sigma_W^2 =0$) if and only if the control-marginal state is pure $\hat \sigma_S = \dyad{w}_S$, and $\ket{w}_S$ is an eigenstate state of the work operator $\hat W_S$.
\end{corollary}
According to the above statement, we can have the following scenarios. First, it is always true for a trivial identity process (since $\hat W_S = 0$), but then $\Delta E_W = \Delta \sigma_W^2 = 0$, i.e., there is no extraction of work neither. Next, for the incoherent work operator $\hat W_S^I$, the zero change of the variance is observed if the initial state is an energy eigenstate, i.e., $\hat \rho_S = \hat \sigma_S = \dyad{\epsilon}_S$. Here, the energy transfer refers to a shift of the weight's energy distribution (the so-called \emph{deterministic work}). Finally, the most interesting is the work extraction from coherences, such that the control-marginal state $\hat \sigma_S = \dyad{w}_S$, but $\ket {w}_S \neq \ket{\epsilon}_S$. It is the problematic case previously encountered in \cite{Allahverdyan2004work} (cf. \cite{Baumer2018}) since the non-zero work can be extracted with zero gain of the variance, and it does not correspond to a simple shift. However, within a framework with explicit energy storage, this cannot be realized in practice exactly since the channel \eqref{control_marginal} is a decoherence process (within the energy basis), such that $\hat \sigma_S$ generally is not a pure state. Only in the limit when weight tends to the time state, i.e., $\hat \rho_W \to \dyad{t}_W$, then  $\hat \rho_S \to \hat \sigma_S = e^{-i\hat H_S t} \hat \rho_S e^{i\hat H_S t}$, and the initial purity is preserved. 

Here, we observe another difference between incoherent and coherent work extraction. In principle, the incoherent work extraction can be deterministic (with $\Delta \sigma_W^2 =0$), and it is independent of the weight state at all. On the contrary, for the work extraction from coherences, we can only consider the limit where $\Delta \sigma_W^2 \approx 0$ and, more importantly, the vanishing variance strongly constrains the initial state of the weight. In particular, to preserve the purity of the control-marginal state, one can consider the limit when weight tends to the time state; however, this specific state also has an infinite variance! The point is that even if we can achieve the small gain of the variance, the final state of the weight would still have substantial energy fluctuations. In other words, the work extraction from coherences depends significantly on the state of the work reservoir (either implicitly through the control-marginal state or explicitly via the $F$-term) and, in particular, one should consider not only the change of the variance but also an absolute (initial or final) dispersion. It is the main subject of the following section.  

\section{Bounds on energy dispersion} \label{bounds_section}
In this section, we want to derive the fundamental bounds for energy dispersion of the energy-storage device within the context of the work extraction process. For this let us define the initial standard deviation $\sigma_E^{(i)}$ and the final one $\sigma_E^{(f)}$ (see Eq. \eqref{variance_time_energy}) of a weight state $\hat \rho_W$ and $\Tr_S[\hat U \hat \rho_S \otimes \hat \rho_W \hat U^\dag]$, respectively. Let us start with the incoherent work extraction process.

\subsection{Incoherent work extraction}
Suppose we consider an incoherent work extraction (either because of the particular incoherent form of the unitary or a presence of the incoherent state). In that case, the control-marginal state is given by $\hat \sigma_S = D[\hat \rho_S]$ and $F = 0$. Then, according to Theorem \ref{work_variance_theorem}, we succeed with the following conclusion.

\begin{corollary} \label{incoherent_work_corollary}
For the incoherent work extraction, i.e., if $\hat W_S = \hat W_S^I$ or $\hat \rho_S = D[\hat \rho_S]$, the extracted work $\Delta E_W$ and change of the variance $\Delta \sigma_W^2$ is independent of the weight state at all. Consequently, there is no fundamental constraint on the initial dispersion $\sigma_E^{(i)}$ and the final dispersion is bounded as follows:
\begin{equation}
    \sigma_E^{(f)} \ge \Var_{\hat \rho_S}[\hat W_S].
\end{equation}
\end{corollary}
The main conclusion from this corollary is that the extraction of an incoherent ergotropy is independent of the initial state of the weight, such that the energy can be stored in a state with a finite energy dispersion. As we will see in the next, this is not true for extracting a coherent part.  

\subsection{Coherent work extraction}
\subsubsection{Ergotropy decoherence} \label{ergotropy_decoherence}
In order to define similar bounds for work extraction involving coherences, first, we want to explain the idea of the ergotropy decoherence. According to Theorem \ref{work_variance_theorem}, work is defined as the ergotropy change of the control-marginal state $\hat \sigma_S$, instead of the proper marginal state $\hat \rho_S$. Essentially, the control-marginal state is the effective density matrix, representing the full statistical knowledge regarding the (average) work extraction. The crucial point is that the map  $\hat \rho_S \to \hat \sigma_S$ \eqref{control_marginal}  is a decoherence process, such that it preserves a diagonal part and decays off-diagonal elements as follows:
\begin{equation} \label{offdiagonal_decay}
    \dyad{\epsilon_i}{\epsilon_j} \to \gamma(\omega_{ij}) \dyad{\epsilon_i}{\epsilon_j}, \ \ |\gamma(\omega_{ij})| \le 1,
\end{equation}
where $\omega_{ij} = \epsilon_j - \epsilon_i$. As it was said, the incoherent part of the ergotropy is unaffected by this process, i.e., $R_I(\rho_S) = R_I(\hat \sigma_S)$, and thus the work extraction from the diagonal is independent of the weight state at all (see Corollary \ref{incoherent_work_corollary}). On the contrary, the coherent part is dumped by the work reservoir, such that $R_C(\hat \rho_S) \ge R_C(\hat \sigma_S)$. That is why we call this phenomenon the \emph{ergotropy decoherence}, which leads to the so-called work-locking \cite{Korzekwa2016, Lobejko2021}, i.e., an inability of work extraction from coherences. The process of ergotropy decoherence \eqref{offdiagonal_decay} is the generalization of the work-locking observed in \cite{Korzekwa2016} (see also \cite{Lostaglio2015, Lostaglio2015_time}), where only the incoherent state of the work reservoir was discussed, and the process is then fully depolarizing (see Eq.  \eqref{weight_incoherent_dephasing}). 

Notice that the dumping function $\gamma(\omega)$ fully characterizes a loss of the coherences and, consequently, the loss of the ergotropy. Although it is straightforward to quantitatively connect the dumping function with some measure of the coherences, it is not so easy to do the same thing with the ergotropy. Recently, some bounds were proposed relating both measures together \cite{Francica2020}. 
 
\subsubsection{Fluctuation-decoherence relation}

Now, we are ready to formulate the fluctuation-decoherence relation, connecting the dumping function $\gamma(\omega)$ with an energy dispersion of the weight $\sigma_E^{(i,f)}$ (initial or final). Firstly, let us notice that $\gamma(\omega)$ is a characteristic function of the time states probability density function (see Eq. \eqref{energy_time_dist}), i.e., according to the definition \eqref{control_marginal} we have:
\begin{equation} \label{gamma_function}
    \gamma(\omega) = \int dt \ g(t) \ e^{i\omega t}.
\end{equation}
On the other hand, time states and energy states are canonically conjugate, such that they satisfy the Heisenberg uncertainty relation (HUR):
\begin{equation} \label{HUR}
    \sigma_t \sigma_E \ge \frac{1}{2},
\end{equation}
for arbitrary state $\hat \rho_W$, where $\sigma_t$ and $\sigma_E$ are square roots of the introduced variances \eqref{variance_time_energy}. Finally, according to the general property of a characteristic function, we have:
\begin{equation} \label{gamma_derivative}
    \sigma_t^2 = -\frac{d^2}{d\omega^2} |\gamma(\omega)| \Big{|}_{\omega=0},
\end{equation}
since $g(t)$ is a probability density function with dispersion $\sigma_t$. 
It proves the relation between the dumping function $\gamma(\omega)$, related to the ergotropy decoherence, and the initial energy dispersion in terms of the HUR \eqref{HUR}. However, this relation is implicit, and the formula \eqref{gamma_derivative} only involves the behavior of a characteristic function close to the origin. 

In the following, we make it more practical, such that the dumping function $\gamma(\omega)$ will define a lower bound for the initial or final energy dispersion. To achieve that goal we need to bound from above the characteristic function by a function that is monotonically decreasing with a dispersion $\sigma_t$, i.e., $|\gamma(\omega)| \le h(\omega, \sigma_t)$, such that $\frac{\partial}{\partial \sigma_t}h(\omega, \sigma_t) < 0$. Then, using HUR we would obtain $|\gamma(\omega)| \le h(\omega, \sigma_t) \le h(\omega, \frac{1}{2 \sigma_E})$, and after optimization over $\omega$ the expected bound would be attained. In the mathematical literature, one can find results for lower and upper bounds for a characteristic function, but none of them is in the required form \cite{Ushakov1997}. Fortunately, recently it was derived the uncertainty relation for the characteristic functions (ChUR) in the form \cite{Rudnicki2016}:  
\begin{equation}
\begin{split}
    &|\gamma(\omega_t)|^2 + |\lambda(\omega_E)|^2 \le \beta(\omega_t \omega_E), \\
    & \beta(x) = 2 \sqrt{2} \frac{\sqrt{2} - \sqrt{1 - \cos x}}{1+\cos x},
\end{split}
\end{equation}
where 
\begin{equation}
    \lambda(\omega) = \int dE \ e^{i \omega E} f(E)
\end{equation}
is an analogous characteristic function for the energy distribution $f(E)$. Using ChUR instead of HUR, we achieve the mentioned goal in the form of the following lemma. 

\begin{lemma} \label{sigmaE_inequality}
For arbitrary $\omega>0$, the following inequality is satisfied:
\begin{equation}
    \sigma_E \ge \frac{\omega |\gamma(\omega)|}{\pi}.
\end{equation}
\end{lemma}

\begin{proof}
Let us introduce two arbitrary real variables $\omega$ and $x$, such that from ChUR we have: 
\begin{equation}
|\gamma(\omega)|^2 \le \beta(x) - |\lambda(x/\omega)|^2.
\end{equation}
Next, we incorporate a lower bound for the modulus of the characteristic function \cite{Ushakov1997, Rudnicki2016}: 
\begin{equation}
|\lambda(\omega)|^2 \ge 1 - \sigma_E^2 \omega^2  
\end{equation}
which brings us the formula:
\begin{equation} \label{pre_lemma_bound}
%|\gamma(\omega_t)|^2 \le f(\omega_t \omega_E) + \sigma_E^2 \omega_E^2 - 1
\sigma_E^2 \ge \frac{\omega^2}{x^2} \left(1- \beta(x) + |\gamma(\omega)|^2 \right).
\end{equation}
Finally, taking the limit $x \to \pi$, we have $f(x) \to 1$, such that
\begin{equation} \label{lemma_bound}
    \sigma_E^2 \ge \frac{\omega^2 |\gamma(\omega)|^2}{\pi^2}.
\end{equation}
\end{proof}

Now, let us come back to the work extraction process $\hat \rho_S \otimes \hat \rho_W \to \hat U \hat \rho_S \otimes \hat \rho_W \hat U^\dag$, for which we define the initial $\sigma_E^{(i)}$ and final $\sigma_E^{(f)}$ standard deviations of the energy distribution. Notice that despite the change of the energy distribution $f_i(E) \to f_f(E)$, the time states probability density function is conserved via the unitary \eqref{unitary_with_S}, namely
\begin{equation}
    g(t) = \Tr[\hat \rho_W \dyad{t}_W] = \Tr[\hat U \hat \rho_W \hat U^\dag \dyad{t}_W], 
\end{equation}
since $\hat S$ and $\hat V_S$ commutes with $\dyad{t}_W$ (see Eq. \eqref{S_operator}). Consequently, the characteristic function $\gamma(\omega)$ is invariant under the work extraction process. 

According to the Lemma \ref{sigmaE_inequality}, this leads us to the second main theorem, defining the lower bound for the work reservoir's initial and final energy fluctuations in terms of the dumping function of the ergotropy. 

\begin{theorem} \label{fluctuation_decoherence_theorem}
The initial $\sigma_E^{(i)}$ and final $\sigma_E^{(f)}$ dispersion of the weight energy satisfy:
\begin{equation}
    \sigma_E^{(i,f)} \ge \frac{1}{\pi} \max_{\omega > 0} \big[ \omega |\gamma(\omega)| \big],
\end{equation}
where $\gamma(\omega)$ is the dumping function of the ergotropy decoherence process.
\end{theorem}
\begin{remark}
Notice that, in general, one could obtain a better bound if instead of inequality \eqref{lemma_bound} it is optimized (over $x$ and $\omega$) the right-hand side of the inequality \eqref{pre_lemma_bound}. However, since our primary goal is to present the general idea of the interplay between fluctuations and dumping of coherences, we use here the simplified formula.     
\end{remark}

Theorem \ref{fluctuation_decoherence_theorem} reveals that the lower bound for (initial or final) dispersion of the weight energy depends on how fast the characteristic function $\gamma(\omega)$ vanishes for large $\omega$, where a slower decay results in higher fluctuations. On the other hand, the $\gamma(\omega)$ is precisely the dumping function of the off-diagonal element for frequency $\omega$. Hence, we observe a trade-off between maintaining the coherent ergotropy and final energy fluctuations. In particular, to unlock more ergotropy, the dumping function has to be close to one for all frequencies appearing in the spectrum of the Hamiltonian $\hat H_S$. But from a property of a characteristic function, if exist $\omega_0 > 0$ such that $|\gamma(\omega_0)| = 1$, then the function is identically equal to one (i.e., for all $\omega$), which implies $\max_{\omega}[ \omega |\gamma(\omega)|] = \infty$. Consequently, we have the following corollary.  

\begin{corollary} \label{dispersion_divergence_corollary}
If the total ergotropy is unlocked, i.e., $R(\hat \rho_S) = R(\hat \sigma_S)$, then $\sigma_E^{(i,f)} = \infty$. 
\end{corollary}

The basic idea of the presented here fluctuation-decoherence relation is that in the function $\gamma(\omega)$ (which represents the initial state of the weight $\hat \rho_W$) is encoded both the information about the locked (coherent) ergotropy of the system and the initial and final dispersion of the energy. Specifically, it provides knowledge about the maximal extractable work and the minimal value of its fluctuations.

\section{Exact Solution: Qubit} \label{qubit_section}
This section explicitly illustrates the presented results by analyzing a two-dimensional system (a qubit) interacting with the weight. In particular, we derive a phase space of possible combinations of the energy changes $\Delta E_W$ and variance changes $\Delta \sigma_W^2$ that can be observed for the arbitrary energy-preserving and translationally-invariant work extraction protocol. 

\subsection{Work-variance phase space}

\begin{figure}[h]
    \centering
    \includegraphics[height = 0.31 \textwidth] {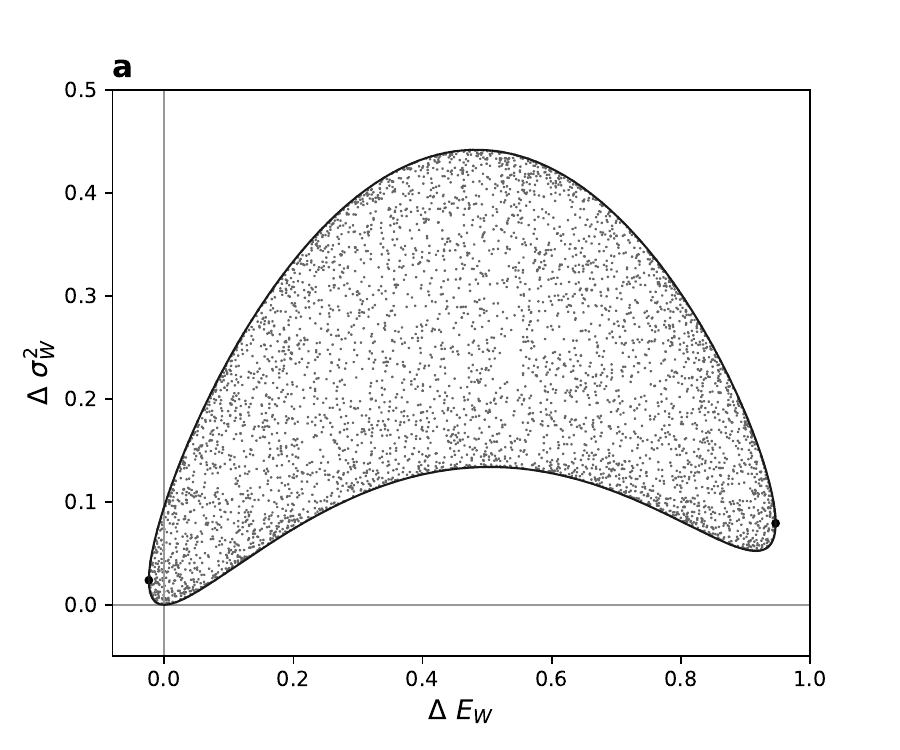}
    \includegraphics[height = 0.31 \textwidth] {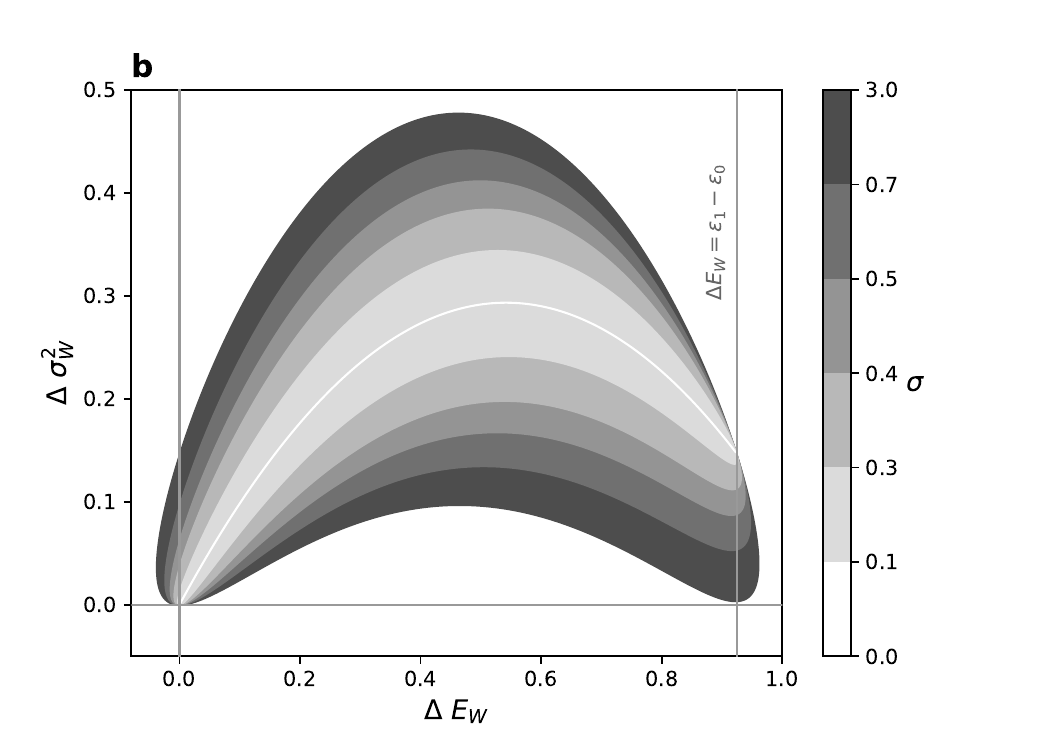}
    \includegraphics[height = 0.31 \textwidth] {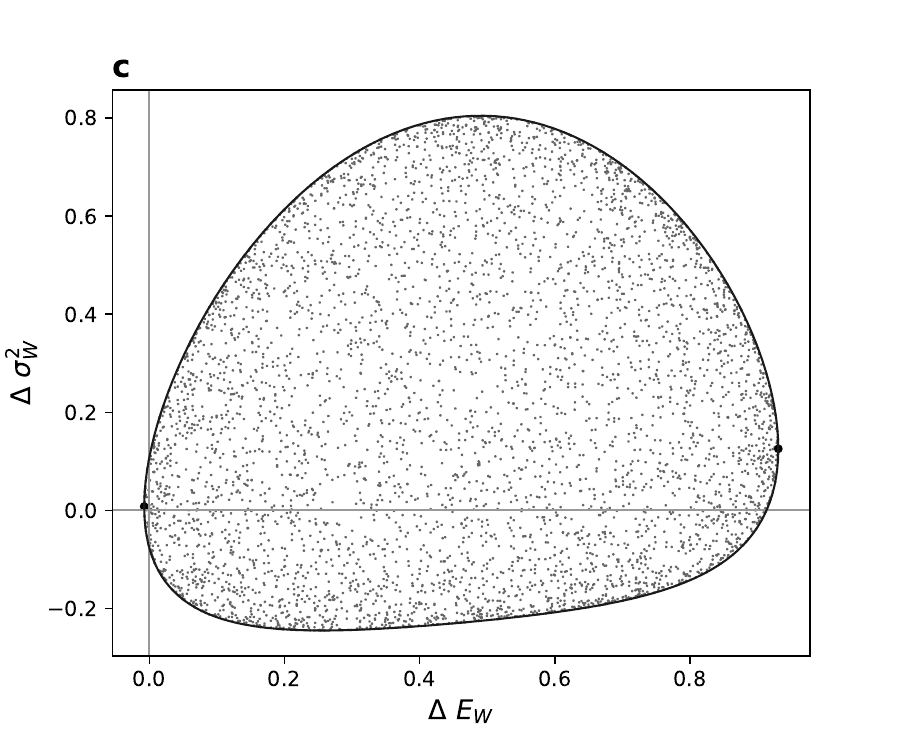}
    \includegraphics[height = 0.31 \textwidth] {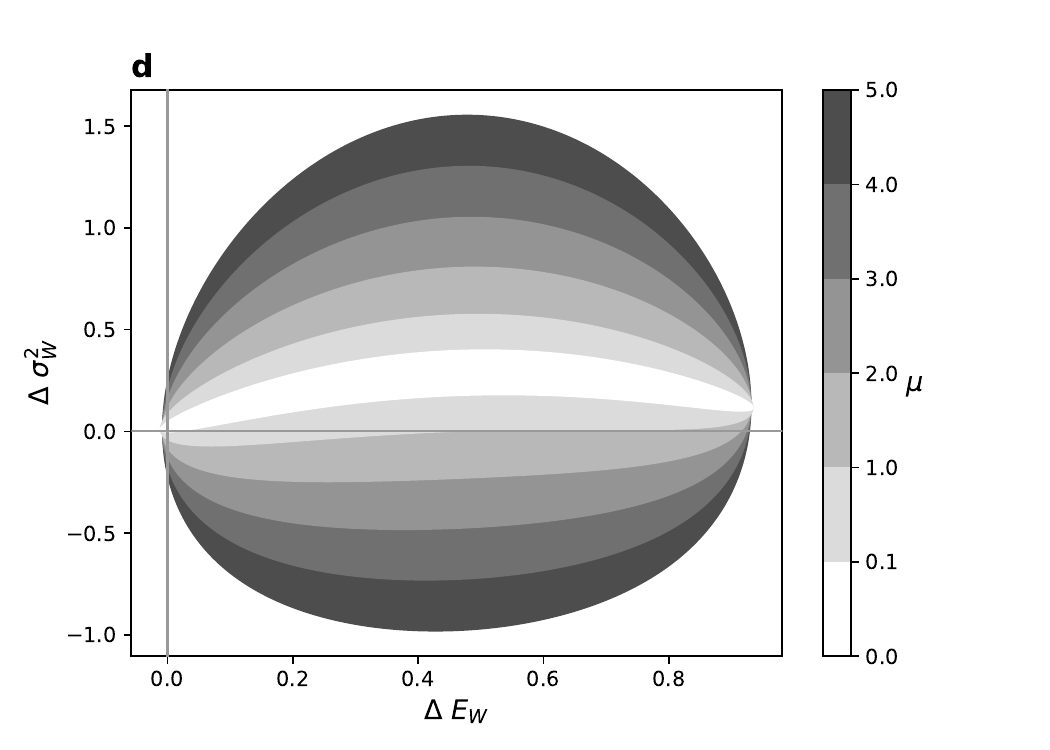} 
    \caption{\emph{Work-variance phase space.} Sets of possible points $(\Delta E_W, \Delta \sigma_W^2)$ for a pure state of the system $\hat \rho_S = \dyad{\psi}_S$, $\ket{\psi}_S \sim \ket{0}_S + 5 \ket{1}_S $. Top panel (\textbf{a}, \textbf{b}) corresponds to a Gaussian wave packet of the weight $\psi_{0, 0, \sigma}$ (semi-classical state) \eqref{gaussian_states}, and a bottom panel (\textbf{c}, \textbf{d}) corresponds to a cat state $\phi_{\mu, 1}$ (non-classical state) \eqref{cat_state}. In the left figures (\textbf{a}, \textbf{c}) (made for a particular values of $\sigma=1/\sqrt{2}$ and $\mu=2$), we plot boundaries of the set coming from Proposition \ref{qubit_proposition}, and we numerically sample over different unitaries $\hat V_S$ (points). The extreme points of the set, corresponding to the minimal and maximal extracted work, are marked by the black dots. For semi-classical state the variance gain is always non-negative, i.e., the set is tangential and above the line $\Delta \sigma_W^2= 0$, whereas for non-classical states there exist a subset with negative changes of the variance. In the plot \textbf{b}, we present boundaries of the phase space for different values of a dispersion $\sigma$. We observe that for $\sigma \to 0$, i.e., in the limit of an incoherent state of the weight, the set reduces to the line. Oppositely, for large values of $\sigma$, the control-state converges to the initial one, i.e., $\hat \sigma_S \to \hat \rho_S$, and due to the initial purity, the set touches the line $\Delta \sigma_W^2 = 0$ for a non-zero work $w = \varepsilon_1 - \varepsilon_0$ (see Eq. \eqref{roots_of_work}).  In the subfigure \textbf{d}, we present the expansion of a phase space with growing value of $\mu$ of the cat state, where $2\mu$ corresponds to a distance between its peaks. Results are presented in energy units given by a multiple of the qubit's gap $\omega$.
    }
    \label{phase_space}
\end{figure}

Let us consider an arbitrary product state of the qubit and the weight $\hat \rho_S \otimes \hat \rho_W$, with the Hamiltonian $\hat H_S = \omega \dyad{1}_S$. The energy gap of the qubit defines a natural energy scale, thus, throughout this section, we work with dimensionless quantities with energetic units given as a multiple of a frequency $\omega$. 

We start with a diagonal representation of the control-marginal state:
\begin{equation}
    \hat \sigma_S = p \dyad{\psi_0}_S + (1-p) \dyad{\psi_1}_S.
\end{equation}
Without loss of generality we assume that $p\le \frac{1}{2}$ and $\langle \hat H_W \rangle_{\hat \rho_W} = 0$. Next, we introduce the (dimensionless) energies: 
\begin{eqnarray}
\varepsilon_i = \frac{1}{\omega} \bra{\psi_i} \hat H_S \ket{\psi_i},
\end{eqnarray}
(for $i=0,1$) such that 
\begin{eqnarray} \label{sum_of_epsilons}
\varepsilon_0 + \varepsilon_1 = \frac{1}{\omega} \Tr[\hat H_S (\dyad{\psi_0}_S + \dyad{\psi_1}_S)] = 1,
\end{eqnarray}
and the following integrals:
\begin{equation}
    \begin{split}
        \eta &= \frac{1}{\omega} \int dt \int dE \ E \ W(E,t) \ (e^{i \omega t} \frac{1}{\gamma} + e^{-i \omega t} \frac{1}{\gamma^*}), \\
        \xi  &= \frac{1}{\omega} \int dt \int dE \ E \ W(E,t) \ (e^{i \omega t} \frac{\varepsilon_1}{\gamma } - e^{-i \omega t} \frac{\varepsilon_0}{\gamma^*} ), \\
        \gamma &= \int dt \int dE \ W(E,t) \ e^{i \omega t}.
    \end{split}
\end{equation}
Then, we propose the following classification of the work extraction protocol:    
\begin{proposition} \label{qubit_proposition}
For an arbitrary work extraction protocol $\hat \rho_S \otimes \hat \rho_W \to \hat U \hat \rho_S \otimes \hat \rho_W \hat U^\dag$ (where $\hat U = \hat S^\dag \hat V_S \hat S$), the corresponding change of the weight's energy and variance $(\Delta E_W, \Delta \sigma_W^2)$ belongs to a set:  
\begin{equation} \label{work_variance_set}
\begin{split}
    \Delta E_W &= w, \ w \in [-\varepsilon_0 (1-2p), \varepsilon_1 (1-2p)], \\
    \Delta \sigma_W^2 &\in \left[f(w) - h(w), f(w) + h(w) \right],
\end{split}
\end{equation}
where 
\begin{equation}
        f(w) = - w^2 + \left(\frac{\varepsilon_1 - \varepsilon_0}{1-2p} + 4 \varepsilon_0 \varepsilon_1 \eta \right) w
        + 2 \varepsilon_0\varepsilon_1 \Big[1 - (1-2p) (\varepsilon_1 - \varepsilon_0) \eta \Big]
\end{equation}
and
\begin{equation}
\begin{split}
    h(w) &= 2 R \sqrt{\varepsilon_0\varepsilon_1 (\varepsilon_0 + \frac{w}{1-2p})(\varepsilon_1 - \frac{w}{1-2p})}, \quad R = \left|1- 2(1-2p) \xi\right|.
\end{split}
\end{equation}
Conversely, for an arbitrary point $(\Delta E_W, \Delta \sigma_W^2)$ within a set \eqref{work_variance_set}, there exist a protocol (i.e., unitary $\hat V_S$) leading to corresponding changes of the weight's cumulants. 
\end{proposition}

In the following, we compare the phase space of possible values $(\Delta E_W, \Delta \sigma_W^2)$ of the semi-classical and non-classical states. The graphical illustration is presented in Fig. \ref{phase_space}. 

\subsubsection{Semi-classical states}
For semi-classical states, the Proposition \ref{qubit_proposition} significantly simplifies since from the condition $\hat \xi'(s) = 0$ (and the assumption $\Tr[\hat H_W \hat \rho_W] = 0$) one can easily show that: 
\begin{equation}
    \int dt \int dE \ E \ e^{i \omega t} \ W(E,t) = 0,
\end{equation}
which further implies $\eta = \xi = 0$. Hence, characterization of semi-classical phase space is given by the functions: 
\begin{equation} \label{semi_classical_phase_space}
    \begin{split}
        f(w) &= - w^2 + \frac{\varepsilon_1 - \varepsilon_0}{1-2p} w + 2 \varepsilon_0\varepsilon_1, \\
        h(w) &= 2 \sqrt{\varepsilon_0 \varepsilon_1 (\varepsilon_0 + \frac{w}{1-2p})(\varepsilon_1 - \frac{w}{1-2p})}.
    \end{split}
\end{equation}

As it follows from Corollary \ref{positive_variance} due to $\Var_{\hat \sigma_S}[\hat W_S] \ge 0$ we have here $f(w) \ge h(w)$ (see Fig. \ref{phase_space}a). Next, we consider the solutions of the equation $f(w_0) = h(w_0)$, which correspond to the energy transfer $w_0$ with a null change of the variance (i.e., $\Delta \sigma_W^2 = 0$). We get the following roots:
\begin{equation} \label{roots_of_work}
    w_0 = 
    \begin{cases}
    0, &  \frac{1}{2} \ge p > 0\\
    0, \varepsilon_1 - \varepsilon_0, & p=0
\end{cases}
\end{equation}
As it was discussed in the previous section, the only positive solution $w_0 = \varepsilon_1 - \varepsilon_0$ is for a pure state ($p=0$) and corresponds to the protocol $\hat V_S$ such that the control-marginal state $\hat \sigma_S$ is the eigenstate of the work operator $\hat W_S$. In Fig. \ref{phase_space}b we present how the point $(\Delta E_W = \varepsilon_1 - \varepsilon_0, \Delta \sigma_W^2 = 0)$ is achieved for a Gaussian wave packet of the weight \eqref{gaussian_states} with an increasing width $\sigma$ (see also comment below Corollary \ref{corollary_zero_variance}). 

Notice that the extreme point $w_{max} = \varepsilon_1 (1-2p)$, for which $h(w_{max}) = 0$, corresponds to the maximal extracted work, i.e., the ergotropy of the state $\hat \sigma_S$. The variance gain if the maximal work is extracted is equal to:
\begin{equation}
    \Delta \sigma_W^2 = f(w_{max}) =  4p (1 - p)\varepsilon_1^2+\varepsilon_0\varepsilon_1.
\end{equation}

\emph{Incoherent states.} Incoherent states form a particular subclass of semi-classicals with $\hat \sigma_S = D[\hat \rho_S]$. Consequently, $\varepsilon_0 = 0$ and $\varepsilon_1 = 1$ (i.e., the eigenstates $\ket{\psi_i}$ are equal to the energy eigenstates $\ket{i}$), such that $h(w)=0$ and the set $(\Delta E_W, \Delta \sigma_W^2)$ reduces to the line:
\begin{equation} \label{parabola}
    \Delta \sigma_W^2 = f(w) = \frac{w}{1-2p} - w^2,
\end{equation}
i.e., for a particular extracted work $\Delta E_W = w$ we have the unique change of the variance. This is also illustrated in Fig. \ref{phase_space}b in the limit of vanishing dispersion of the wave packet $\sigma \to 0$ (such that wave function tends to the Dirac delta). 

\begin{figure}[ht]
    \centering
    \includegraphics[width = 0.3 \textwidth] {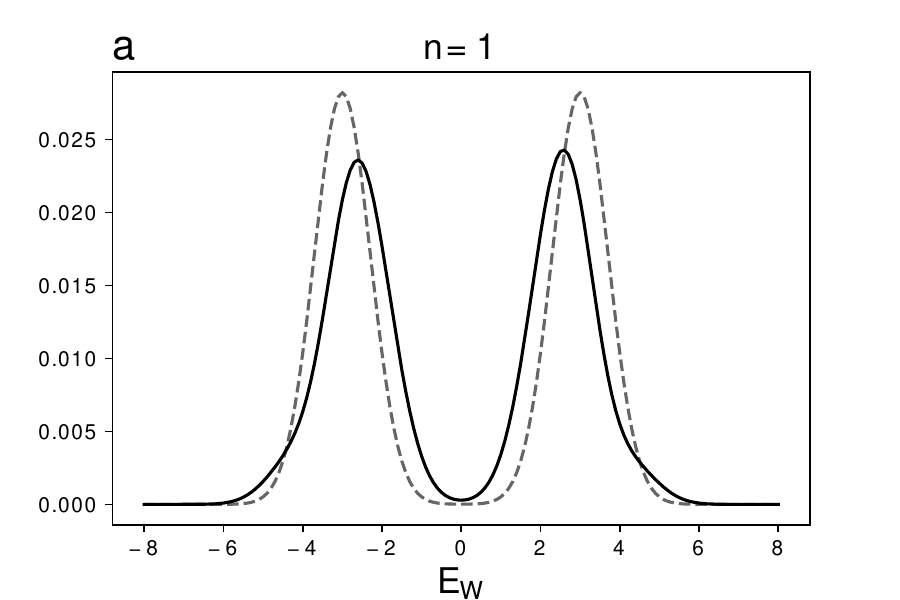}
    \includegraphics[width = 0.3 \textwidth] {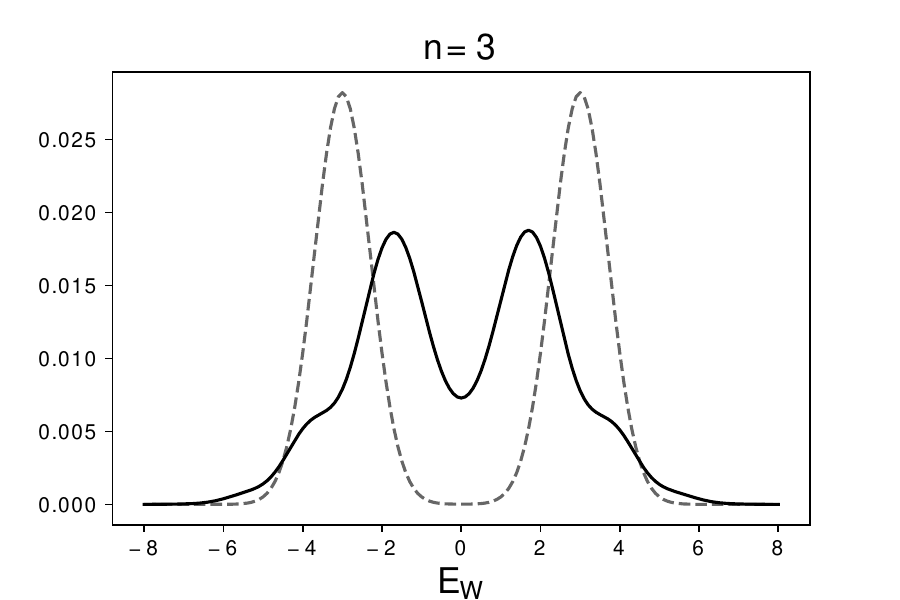}
    \includegraphics[width = 0.3 \textwidth] {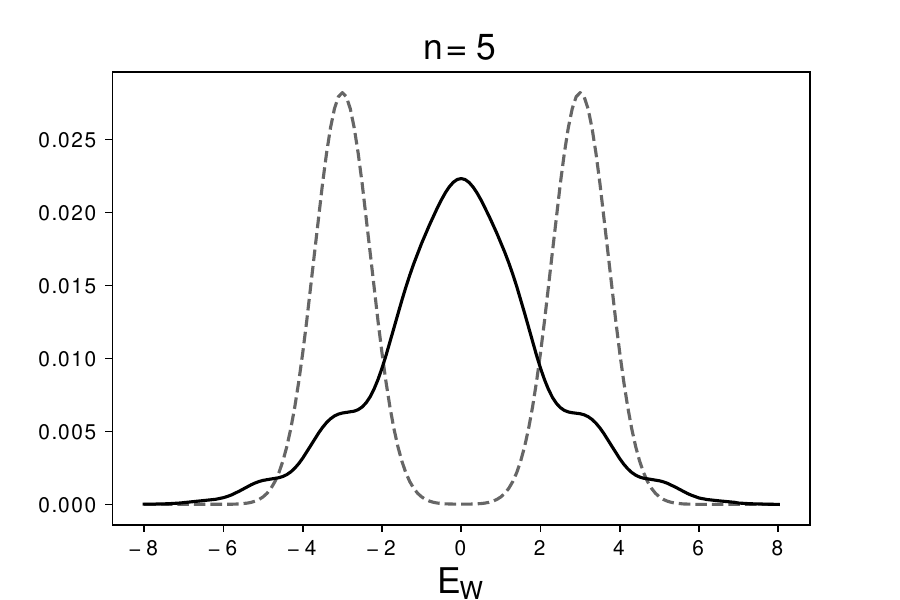}
    \includegraphics[width = 0.3 \textwidth] {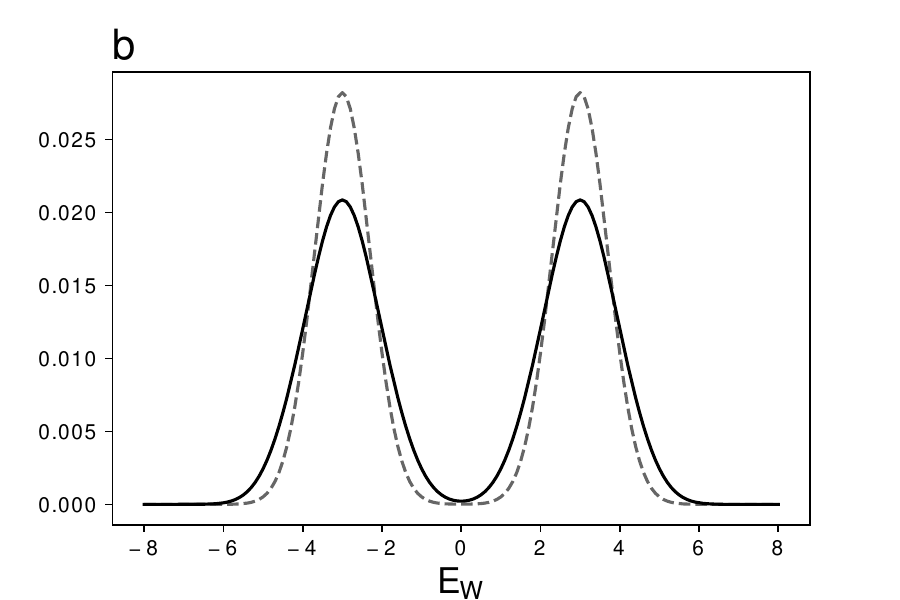}
    \includegraphics[width = 0.3 \textwidth] {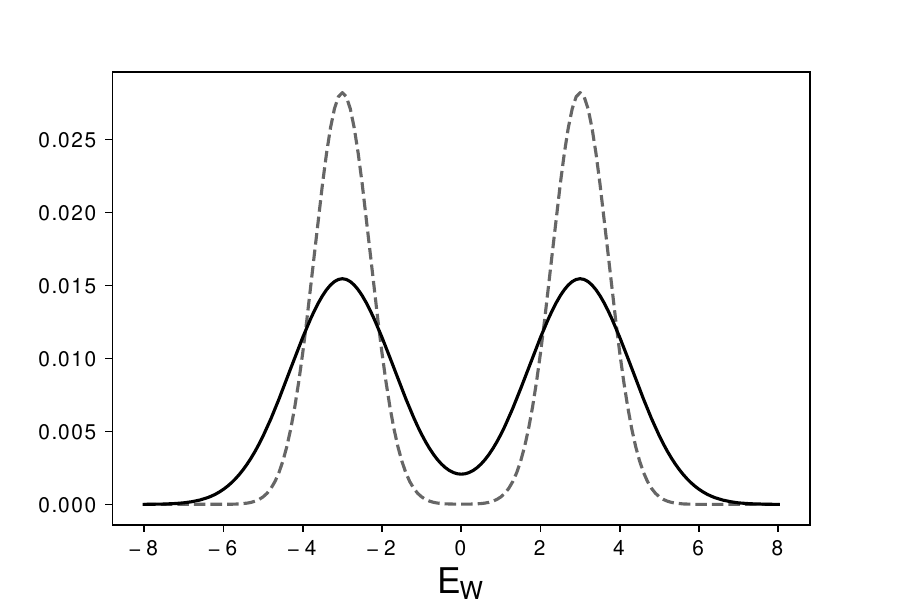}
    \includegraphics[width = 0.3 \textwidth] {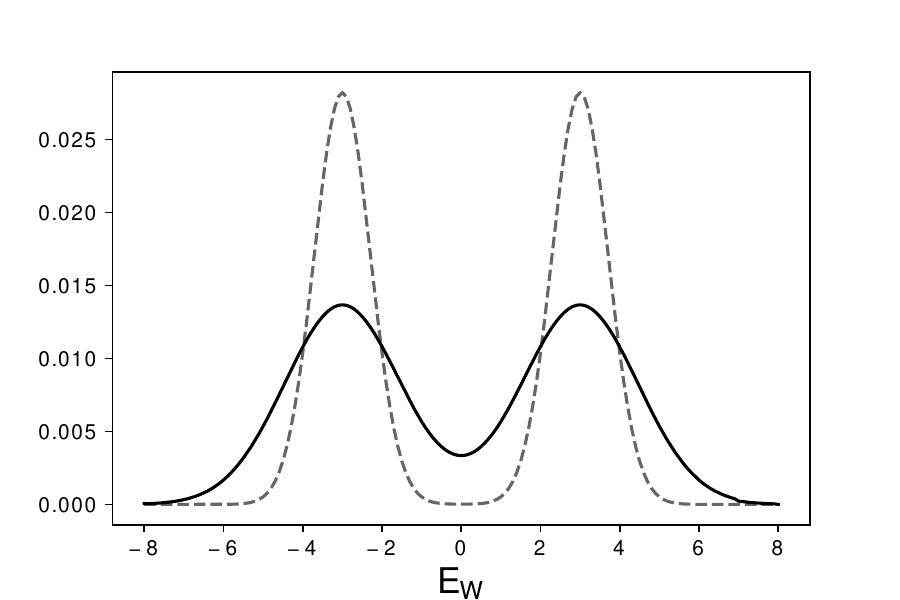}

    \caption{\emph{Subsequent reduction of energy fluctuations.} The energy distribution $f_n(E) = \bra{E} \hat \rho_W^{(n)} \ket E$ of a weight state obtained through a subsequent coupling with a system, i.e., $\hat \rho_W^{(n)} = \Tr_S[\hat U_n \hat \rho_S \otimes \hat \rho_W^{(n-1)} \hat U_n^\dag]$, where the initial state $\hat \rho_W^{(0)} = \dyad{\phi}_W$ is a cat state with a wave function $\phi_{3,1}$ \eqref{cat_state}. Top panel (\textbf{a}) is for the coherent ``plus'' state $\hat \rho_S = \dyad{+}_S$ \eqref{plus_state}, and bottom panel (\textbf{b}) for the incoherent state $\hat \rho_S = \frac{1}{2} (\dyad{0}_S + \dyad{1}_S)$. In both cases the evolution operator $\hat U_n = \hat S^\dag \hat V^{(n)}_S \hat S$ is for a unitary $\hat V_S^{(n)}$ which minimizes the variance gain given by Eq. \eqref{variance_minimum}. The solid line represents the probability density function $f_n(E)$ for $n=1, 3, 5$, whereas the dashed line is a reference distribution for $n=0$. It is seen that for a coherent state (\textbf{a}), the quantum interference leads to sequential reduction of the weight's energy variance via collapsing to each other both peaks of the cat state. On the contrary, if the system is incoherent (\textbf{b}), the process only leads to broadening of each peaks, and hence to increasing of the energy variance, as it is expected for an arbitrary semi-classical state. Results are presented in energy units given by a multiple of the qubit's gap $\omega$. }
    \label{reducing_variance}
\end{figure}

\subsubsection{Non-classical states with reflection symmetry}
As an example of the non-classical states, we consider the weight's wave function with reflection symmetry, i.e. $\psi(E) = \psi(-E)$, for which the Wigner function is symmetric : $W(E,t) = W(-E, -t)$. Due to this symmetry, we observe (by a simple change of the variables) that the dephasing factor $\gamma$ is a real number, and hence $\eta = 0$. 

Next, one should notice that applying the relation $\varepsilon_1 = 1 - \varepsilon_0$ one can rewritten the $\xi$ function in the form:
\begin{equation}
\begin{split}
\xi  &= \frac{1}{\gamma \omega} \int dt \int dE \ E \ W(E,t) \ e^{i \omega t} - \varepsilon_0 \eta,
\end{split}
\end{equation}
what finally gives us the expression:
\begin{equation}
    \xi = \frac{1}{\omega} \frac{\int dt \int dE \ E \ e^{i \omega t} \ W(E,t) }{\int dt \int dE \ e^{i \omega t} \ W(E,t)}.
\end{equation}
Finally, one should observe that $\xi$ is purely imaginary, i.e., $\xi^* = -\xi$, such that we have
\begin{equation}
    R = \sqrt{1 + 4 (1-2p)^2 |\xi|^2}.
\end{equation}
We see that for states with reflection symmetry the phase space of the variables $(\Delta E_W, \Delta \sigma_W^2)$ is exactly the same as for the semi-classical states given by Eq. \eqref{semi_classical_phase_space}, but with the radius $R\ge1$ (where for semi-classical states $R=1$). 

\emph{Cat states.} As a particular example of space-symmetric states, we take the so-called ``cat state'', with the wave function:  
\begin{equation} \label{cat_state}
    \phi_{\mu,\nu} (E) = \frac{\psi_{\mu,\nu, \frac{1}{\sqrt{2}}} (E) + \psi_{-\mu,-\nu, \frac{1}{\sqrt{2}}} (E)}{\sqrt{2(1+e^{-\mu^2 - \nu^2})}},  
\end{equation}
where $\psi_{\mu,\nu, 1/\sqrt{2}} $ is a Gaussian wave packet \eqref{gaussian_states} with $\sigma = 1/\sqrt{2}$. For this we derive an exact analytical formula for the radius $R$, i.e.:
\begin{equation}
\begin{split}
    %R &= \sqrt{1+\frac{4(1-2p)^2\left(\nu (1 - e^{2 \mu \omega})- 2 \mu e^{\nu^2 + \mu^2 + \mu \omega}\sin(\nu \omega)\right)^2}{\omega^2 \left(1+ e^{2 \mu \omega} + 2 e^{\nu^2 + \mu^2 + \omega \mu} \cos(\nu \omega)\right)^2}} \\
    % &R = \sqrt{1+\frac{4(1-2p)^2\left(\nu (1 - e^{2 \mu \omega})- 2 \mu \kappa_{\mu, \nu} \sin(\nu \omega)\right)^2}{\omega^2 \left(1+ e^{2 \mu \omega} + 2 \kappa_{\mu, \nu} \cos(\nu \omega)\right)^2}}, \\
    % &\kappa_{\mu, \nu} = e^{\nu^2 + \mu^2 + \mu \omega}.
    &R = \sqrt{1+\frac{4(1-2p)^2\left( \nu (1 - e^{2  \mu})- 2  \mu \kappa_{ \mu,  \nu} \sin( \nu)\right)^2}{ \left(1+ e^{2  \mu} + 2 \kappa_{ \mu,  \nu} \cos( \nu)\right)^2}}, \\
    &\kappa_{ \mu,  \nu} = e^{ \nu^2 +  \mu^2 +  \mu}.
\end{split}
\end{equation}
One should notice that $R=1$ either if $\mu = 0$ or $\nu = 0$. 

\emph{Coherent ``plus" state.} We would like to derive a phase space for a coherent initial state of the system $\hat \rho_S = \dyad{+}_S$, where
\begin{equation} \label{plus_state}
\ket{+}_S = \frac{1}{\sqrt{2}} (\ket{0}_S + \ket{1}_S).    
\end{equation}
For this we need to show that $\varepsilon_0 = \varepsilon_1 = 1/2$. Indeed, the transition $\hat \rho_S \to \hat \sigma_S$ is a decoherence process, thus  $\langle \hat H_S \rangle_{\hat \sigma_S} = \langle \hat H_S \rangle_{\hat \rho_S}$ for arbitrary state of the weight. Then,  we have a set of equalities (see Eq. \eqref{sum_of_epsilons}):
\begin{align*}
        &p \varepsilon_0 + (1-p) \varepsilon_1 = \frac{1}{2}, \\
        &\varepsilon_0 +  \varepsilon_1 = 1.
\end{align*}
However, this should be satisfied for any value of $p \le \frac{1}{2}$ (which depends on the state $\hat \rho_W$), thus the only possible solution is $\varepsilon_0 = \varepsilon_1 = 1/2$. Finally, we put it to Proposition \ref{qubit_proposition}, and get:
\begin{equation} \label{coherent_phase_space}
\begin{split}
    f(w) &= - w^2 + \eta w + \frac{1}{2}, \\
    h(w) &= R \sqrt{\frac{1}{4} - \frac{w^2}{(1-2p)^2}}. \\
\end{split}
\end{equation}

\emph{Reducing of the energy dispersion.} The most interesting feature of non-classical states is the possibility of reducing the variance in the final weight state. We analyze this process of squeezing for a particular coherent initial state $\hat \rho_S = \dyad{+}_S$ (introduced above) and pure state of the weight with reflection symmetry (with $\eta = 0$ and $R \ge 1$). From Eq. \eqref{coherent_phase_space} one can calculate the maximal possible drop of the variance, which is achieved for the point $w = 0$, i.e., when no work is extracted. The minimum is given by: 
\begin{equation} \label{variance_minimum}
    \min[\Delta \sigma_W^2] = f(0) - h(0) = \frac{\omega}{2} (1 - R).
\end{equation}
We observe that for $R > 1$, the change of the variance is negative. 

Fig. \ref{reducing_variance} presents how the cat state of the weight changes according to the process of such variance reduction. Moreover, we show that one can decrease the variance several times if subsequent protocols with the same state $\hat \rho_S$ are introduced. We compare it with the same evolution but with the fully dephased state $D[\hat \rho_S]$ (i.e., without coherences). It is seen that as long as the classical state only broadens both peaks of the cat state (which is a universal feature of all semi-classical states, see Corollary \ref{positive_variance}), a quantum interference can collapse those peaks to each other, such that the energy dispersion is reduced. However, we stress that such a subsequent reduction eventually saturates and cannot be repeated forever. Indeed, we have previously shown that a Gaussian wave packet also belongs to the semi-classical states; hence, no squeezing of its width is possible. It suggests that the presence of two peaks (i.e., with $\mu > 0$) is essential for the process of the variance reduction, but still, the reason for that is a quantum interference between the qubit and the weight. Especially, as it was discussed before, the cat state has $R > 1$, and hence $\Delta \sigma_W^2 < 0$, only if $\mu > 0$ and $\nu > 0$.   

\subsection{Coherent vs. Incoherent work extraction: Gaussian state and qubit}
In the section \ref{bounds_section} we discussed the bounds for an energy dispersion if the incoherent (coherent) part of ergotropy is extracted. Let us illustrate this behavior quantitatively for a particular example.

We consider an arbitrary Gaussian wave packet of the weight in the form \eqref{gaussian_states} and a system given by a qubit in a state:
\begin{eqnarray} \label{qubit_state}
\hat \rho_S = \frac{1}{2} (\mathbb{1} + \vec x \cdot \vec{\hat \sigma}),
\end{eqnarray}
where $\vec x = (x, y, z)$ and $\vec {\hat \sigma}$ is a vector of Pauli matrices. We are interested in the final dispersion $\sigma_E^{(f)}$ of the weight after a protocol, if either is extracted the maximal incoherent work given by $R_I \equiv R_I(\hat \rho_S)$ or the coherent ergotropy $R_C \equiv R_C(\hat \rho_S)$.

For a qubit, we have the following expressions (in units of the energy gap $\omega$):
\begin{eqnarray}
    R_I &=& \frac{1}{2} \left(|z| - z \right), \label{incoherent_qubit_ergotropy} \\
    R_C &=& \frac{1}{2} \left(\sqrt{|\gamma(\omega)|^2 \alpha^2 + z^2} - |z|\right) \label{coherent_qubit_ergotropy}.
\end{eqnarray}
where we put $\alpha^2 = x^2 + y^2 $. As it follows from Theorem \eqref{fluctuation_decoherence_theorem}, the dumping factor $\gamma(\omega)$ is the characteristic function of the weight's time states (given by Eq. \eqref{gamma_function}). Hence, we see that the maximal coherent work depends explicitly on the initial state of the energy-storage device. For a Gaussian state, the function $\gamma(\omega)$ takes a simple exponential form (see Appendix \ref{coherent_incoherent_appendix}), which together with HUR relation \eqref{HUR} leads us to the following result: If via the protocol the coherent ergotropy $R_C$ is extracted, the final dispersion of energy is bounded by:
\begin{eqnarray} \label{coherent_bound_qubit}
\sigma_E^{(f)} &\ge& \frac{1}{2 \sqrt{\log \left[\frac{\alpha^2}{4 R_C(R_C + |z|)} \right]} }.
\end{eqnarray}
On the contrary, if the process is solely incoherent (either because of $\alpha = 0$ or $\hat W_S = \hat W_S^I$), then extracting the maximal work $R_I$ results in the bound (see Corollary \ref{incoherent_work_corollary}):
\begin{eqnarray} \label{incoherent_bound_qubit}
\sigma_E^{(f)} &\ge&  \sqrt{1 - R_I^2}.
\end{eqnarray}

For a specific initial state $\hat \rho_S$, the incoherent ergotropy $R_I$, and therefore also the bound \eqref{incoherent_bound_qubit}, is fixed and finite. However, the coherent part lies in the set $R_C \in \left[0, \frac{1}{2} \left(|\vec x|- |z|\right) \right]$ (depending on the initial state of the weight). In Fig. \ref{coherent_bound_plot} we plot how the minimal final dispersion changes with respect to the value of $R_C$ within this set. In particular, it is seen that the dispersion diverges if the maximal coherent work is extracted in accordance with Corollary \ref{dispersion_divergence_corollary}. 

\begin{figure}[ht]
    \centering
    \includegraphics[width = 0.5 \textwidth] {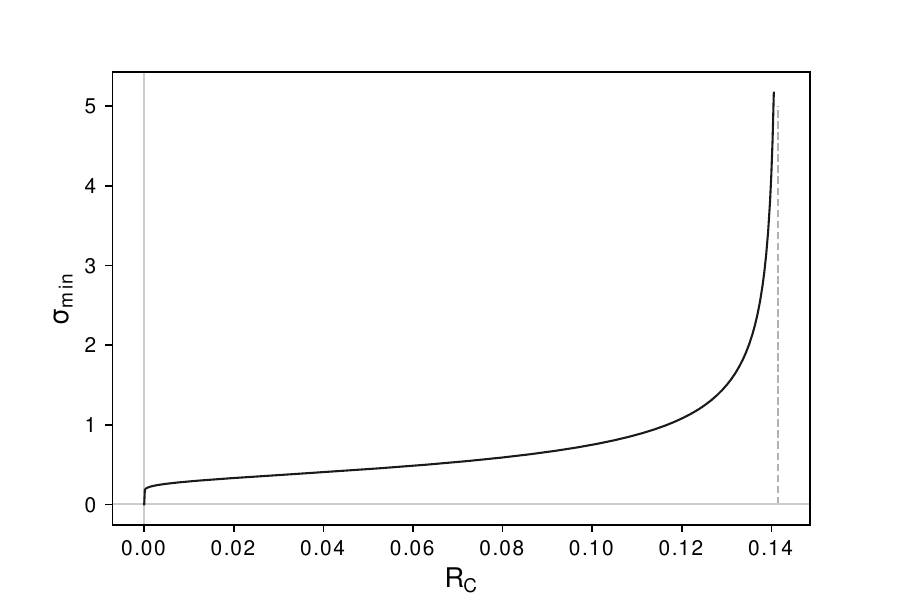}
    \caption{\emph{The lower bound for the energy dispersion.} The function (solid line) represents the minimal final dispersion of the weight for an arbitrary initial Gaussian state \eqref{gaussian_states} and state of a qubit \eqref{qubit_state} with $\alpha = 0.5$ and $z = 0.3$, if the coherent ergotropy $R_C$ \eqref{coherent_qubit_ergotropy} is extracted. The dashed line correspond to the maximal value $R_C = \frac{1}{2} \left(|\vec x| - |z|\right)$ for which the minimal variance diverges. Results are presented in energy units given by a multiple of the qubit frequency $\omega$.}
    \label{coherent_bound_plot}
\end{figure}

\section{Summary and Discussion} \label{summary_section}

In this paper, we contribute to the long-standing problem of micro-scale thermodynamics, i.e., what role play quantum coherences in the process of work extraction? We have shown that to answer this question correctly, it is not enough to construct a quasi-distribution with a proper classical limit but also to introduce an energy-storage device as an independent quantum system (as expected for fully autonomous thermal machines).  

We consider a comprehensive study of work fluctuations within such a fully quantum setup. Primarily, we have shown the significance of an autonomous approach in Theorem \ref{work_variance_theorem}, where the $F$-term, involving interference between the system and the weight, appears as an additional contribution to changes in the energy variance (i.e., representing fluctuations of the extracted work). This contribution is crucial since the $F$-term for quantum systems can lead to a new qualitative phenomenon of reducing work fluctuations, in contrast to semi-classical states, for which the variance always increases. However, we stress that not every state with coherences manifests those features, e.g., coherent Gaussian wave packets, behave as a semi-classical system as expected. 

Besides, we derive the fluctuation-decoherence relation showing that the process of a work-locking caused by decoherence is related to energy dispersion of the work reservoir. In particular, we reveal that unlocking the total coherent ergotropy always results in divergence of the work fluctuations. In general, this observation points out the main difference between the extraction of coherent and incoherent work: The former can decrease the variance, but its absolute value diverges if more and more energy is extracted, whereas for the latter, the gain is always non-negative, but a total (incoherent) ergotropy can be extracted with finite work fluctuations.

Presented here framework opens a bunch of new research areas. Firstly, one can ask what the formulas are for higher cumulants of the quasi-distribution. In particular, identifying, similar to the $F$-term, quantum contributions in higher cumulants shall answer the question of the impact of quantum coherences on the work extraction process. Moreover, in this paper, we only concentrate on product states, remaining an open question of what role quantum correlations like entanglement can play. We also do not address the issue of possible degeneracies in the weight spectrum. Furthermore, the introduced quasi-distribution is interesting on its own; due to its robustness to the invasiveness of the measurements and physical interpretation of cumulants, it can be successfully applied to the analysis of other fluctuation theorems (e.g., involving heat transfer).

Finally, we stress that our framework is a step toward formulating the necessary conditions for a work reservoir diagonal part transition. Indeed, by derived theorems, we are able to characterize a work-variance phase space of a qubit completely; hence a natural extension is to do the same thing for higher dimensions and cumulants. In this sense, a complete characterization of the energy statistics evolution can be understood as an ultimate formulation of the fluctuation theorems.

\section*{Acknowledgements}
The author thanks  Anthony J. Short and Micha{\l} Horodecki for helpful and inspiring discussions.  This research was supported by the National Science Centre, Poland, through grant  SONATINA 2 2018/28/C/ST2/00364.

\bibliographystyle{vancouver}
\bibliography{bib}

\newpage

\appendix

\section{Proof of Theorem \ref{work_variance_theorem}}
\subsection*{Part I ($\Delta E_W$ and $\Delta \sigma_W^2$)}

We want to show the following equalities:
\begin{equation}
\begin{split}
    \Delta E_W &= \langle \hat U^\dag \hat H_W \hat U\rangle_{\hat \rho} - \langle \hat H_W \rangle_{\hat \rho} = \langle \hat W_S \rangle_{\hat \sigma}, \\
    \Delta \sigma_W^2 &= \Var_{\hat \rho}[\hat U^\dag \hat  H_W \hat U] - \Var_{\hat \rho}[\hat H_W] = \Var_{\hat \sigma}[\hat W_S] + 2 \Cov_{\hat \sigma} [\hat H_S - \hat H_W, \hat H_S'],
\end{split}
\end{equation}
where
\begin{eqnarray}
    \hat \rho &= \hat \rho_S \otimes \hat \rho_W, \ \ \hat \sigma = \hat S \hat \rho \hat S^\dag, \ \ \hat S = e^{-i \hat H_S \otimes \hat \Delta_W }, \\  \hat U &= \hat S^\dag \hat V_S \hat S, \ \ \hat W_S = \hat H_S - \hat H_S', \ \ \hat H_S' = \hat V_S^\dag \hat H_S \hat V_S.
\end{eqnarray}
Notice that operators with subscribe $S$ or $W$ acts entirely on the system or the weight Hilbert space, respectively, whereas operators without subscribe acts on both.  

The proof is solely based on the following change of the basis:
\begin{eqnarray}
    \langle \hat A \rangle_{\hat \rho} = \Tr[\hat A \hat \rho] = \Tr[\hat S \hat A \hat S^\dag \hat S \hat \rho \hat S^\dag] = \Tr[\hat S \hat A \hat S^\dag \hat \sigma ] = \langle \hat S \hat A \hat S^\dag \rangle_{\hat \sigma} , \ \ \langle \hat A \rangle_{\hat \sigma} = \langle \hat S^\dag \hat A \hat S \rangle_{\hat \rho},
\end{eqnarray}
alongside with the commutation relations:
\begin{eqnarray}
 [\hat H_S, \hat S] = 0, \ \  [\hat H_W, \hat S] = \hat S \hat H_S.
\end{eqnarray}
The first commutator is obvious, whereas the second comes from the canonical commutation relation: $[\hat H_W, \hat \Delta_W] = i$. Using these we simply derive the following equalities:
\begin{eqnarray}
\langle \hat U^\dag \hat H_W \hat U \rangle_{\hat \rho} &=& \langle \hat H_W \rangle_{\hat \sigma} -\langle \hat H_S' \rangle_{\hat \sigma}, \label{UHW_appendix}\\
\langle \hat U^\dag \hat  H_W^2 \hat U \rangle_{\hat \rho} &=& \langle \hat H_W^2 \rangle_{\hat \sigma} + \langle \hat H_S'^2 \rangle_{\hat \sigma} - 2 \langle \hat H_W \hat H_S' \rangle_{\hat \sigma}, \label{UHW2_appendix} \\
\langle \hat H_W \rangle_{\hat \sigma} &=&  \langle \hat H_W \rangle_{\hat \rho} + \langle  \hat H_S  \rangle_{\hat \rho}, \label{HW_appendix} \\
\langle \hat H_W^2 \rangle_{\hat \sigma} &=&  \langle \hat H_W^2 \rangle_{\hat \rho} + \langle  \hat H_S^2  \rangle_{\hat \rho} - 2 \langle  \hat H_S  \rangle_{\hat \rho} \langle  \hat H_W  \rangle_{\hat \rho}, \label{HW2_appendix} \\
\langle \hat H_S \rangle_{\hat \sigma} &=& \langle \hat H_S \rangle_{\hat \rho},  \ \ \langle \hat H_S^2 \rangle_{\hat \sigma} = \langle \hat H_S^2 \rangle_{\hat \rho}, \ \  \Var_{\hat \rho}[\hat H_S] = \Var_{\hat \sigma}[\hat H_S], \label{HS_appendix} \\
\langle \hat H_S \hat H_W \rangle_{\hat \sigma} &=& \langle \hat H_S^2 \rangle_{\hat \rho} + \langle \hat H_S \rangle_{\hat \rho}\langle \hat H_W \rangle_{\hat \rho}, \label{HSHW_appendix}
\end{eqnarray}
and then we have
\begin{eqnarray}
    \Cov_{\hat \sigma}[\hat H_S, \hat H_W] &=& \langle \hat H_S \hat H_W \rangle_{\hat \sigma} - \langle \hat H_S \rangle_{\hat \sigma} \langle \hat H_W \rangle_{\hat \sigma} \nonumber \\ 
    &\stackrel{(\ref{HW_appendix}, \ref{HS_appendix}, \ref{HSHW_appendix})}{=}& \langle \hat H_S^2 \rangle_{\hat \rho} + \langle \hat H_S \rangle_{\hat \rho}\langle \hat H_W \rangle_{\hat \rho} - \langle \hat H_S \rangle_{\hat \rho} (\langle \hat H_W \rangle_{\hat \rho} + \langle \hat H_S \rangle_{\hat \rho}) = \Var_{\hat \sigma}[\hat H_S] \label{Cov_HS_HW_appendix} \\
    \Var_{\hat \rho}[\hat U^\dag \hat  H_W \hat U] &=& \langle \hat U^\dag \hat  H_W^2 \hat U \rangle_{\hat \rho} - \langle \hat U^\dag \hat  H_W \hat U \rangle_{\hat \rho} \langle \hat U^\dag \hat  H_W \hat U \rangle_{\hat \rho} \nonumber \\
    &\stackrel{(\ref{UHW_appendix}, \ref{UHW2_appendix})}{=}& \Var_{\hat \sigma} [\hat H_W] + \Var_{\hat \sigma} [\hat H_S'] - 2 \Cov_{\hat \sigma}[\hat H_S', \hat H_W] \label{varUHW_appendix} \\
    \Var_{\hat \rho}[\hat  H_W] &=& \langle \hat  H_W^2 \rangle_{\hat \rho} - \langle  \hat  H_W  \rangle_{\hat \rho} \langle  \hat  H_W \rangle_{\hat \rho} \nonumber \\
    &\stackrel{(\ref{HW_appendix}, \ref{HW2_appendix})}{=}& \Var_{\hat \sigma} [\hat H_W] + \Var_{\hat \sigma} [\hat H_S] - 2 \Cov_{\hat \sigma} [\hat H_S, \hat H_W] \label{varHW_appendix}
\end{eqnarray} 
From the definition of $\hat W_S$ we also have
\begin{equation} \label{WS_appendix}
    \Var_{\hat \sigma}[\hat W_S] = \Var_{\hat \sigma}[\hat H_S] + \Var_{\hat \sigma}[\hat H_S'] - 2 \Cov_{\hat \sigma}[\hat H_S, \hat H_S'].
\end{equation}
Finally, we are ready to apply above relations and proof the theorem:
\begin{equation}
\begin{split}
    \Delta E_W &= \langle \hat U^\dag \hat H_W \hat U \rangle_{\hat \rho} - \langle \hat H_W \rangle_{\hat \rho} \stackrel{(\ref{UHW_appendix}, \ref{HW_appendix})}{=} \langle  \hat H_S \rangle_{\hat \sigma} -  \langle \hat H_S' \rangle_{\hat \sigma} = \langle  \hat W_S \rangle_{\hat \sigma} \\
    \Delta \sigma_W^2 &= \Var_{\hat \rho}[\hat U^\dag \hat  H_W \hat U] - \Var_{\hat \rho}[\hat H_W] \\
    &\stackrel{(\ref{varUHW_appendix}, \ref{varHW_appendix})}{=} \Var_{\hat \sigma}[\hat H_S'] - \Var_{\hat \sigma}[\hat H_S] + 2 \Cov_{\hat \sigma}[\hat H_S-\hat H_S', \hat H_W] \\
    &\stackrel{\eqref{WS_appendix}}{=} \Var_{\hat \sigma}[\hat W_S] - 2 \Var_{\hat \sigma}[\hat H_S] + 2 \Cov_{\hat \sigma}[\hat H_S, \hat H_S'] +  2 \Cov_{\hat \sigma}[\hat H_S-\hat H_S', \hat H_W] \\
    &\stackrel{\eqref{Cov_HS_HW_appendix}}{=} \Var_{\hat \sigma}[\hat W_S] + 2 \Cov_{\hat \sigma}[\hat H_S -\hat H_W, \hat H_S'] 
\end{split}
\end{equation}

\subsection*{Part II (equivalent forms of the $F$-term)}
Now, we want to show an equivalence of different forms of the $F$-term, namely
\begin{align} \label{F_appendix}
    F &= \Cov_{\hat \sigma} [\hat H_S - \hat H_W, \hat H_S'] \nonumber \\
    &= \int dt \Tr[\hat W_S e^{-i \hat H_S t} \hat \rho_S e^{i \hat H_S t}] \int dE \ (E -  \langle \hat H_W \rangle_{\hat \rho}) \ W(E,t) = -i \Tr[\hat W_S \hat \xi'_S (0)]
\end{align}
where $\hat \xi(s)$ is given by Eq. \eqref{xi_state}.
Here, we introduce the Wigner function:
\begin{equation}
    W(E,t) = \frac{1}{2\pi} \int d\omega \ e^{i \omega t} \bra{E - \frac{\omega}{2}  } \hat \rho_W \ket{E + \frac{\omega}{2}}.
\end{equation}
and, in particular, we will use the following identities:
\begin{align}
    \rho(E-\frac{\omega_{ij}}{2} , E+\frac{\omega_{ij}}{2}) &= \int dt \ e^{-i\omega_{ij} t} \ W(E,t), \\
    \int dE  \ f(E) \ \rho(E-\varepsilon_i, E-\varepsilon_j) &=  \int dE  f[E + \frac{1}{2} (\epsilon_i + \epsilon_j)] \rho(E-\frac{\omega_{ij}}{2} , E+\frac{\omega_{ij}}{2}) \nonumber \\
    &= \int dE \ f[E + \frac{1}{2} (\epsilon_i + \epsilon_j)] \int dt \ e^{-i\omega_{ij} t} \ W(E,t) \label{Wigner_iden},
\end{align}
where $\omega_{ij} = \epsilon_i - \epsilon_j$ and $\rho(E,E') = \bra{E} \hat \rho_W \ket{E'}$. We also introduce the matrix element $\rho_{ij} = \bra{\epsilon_i} \hat \rho_S \ket{\varepsilon_j}$. Then, using the relation $\hat S \ket{\epsilon_i, E} = \ket{\epsilon_i, E+\epsilon_i}$, the control state is equal to:
\begin{equation} \label{control_state_appendix}
\begin{split}
    \hat \sigma = \hat S \hat \rho_S \otimes \hat \rho_W \hat S^\dag &= \int dE dE' \sum_{i,j} \rho_{ij} \dyad{\epsilon_i}{\epsilon_j} \otimes \rho(E, E') \dyad{E+\epsilon_i}{E'-\epsilon_j} \\
    &= \int dE dE' \sum_{i,j} \rho_{ij} \dyad{\epsilon_i}{\epsilon_j} \otimes \rho(E-\epsilon_i, E'-\epsilon_j) \dyad{E}{E'}.
\end{split}
\end{equation}
According to this representation, we have:
\begin{equation} \label{VS_HS_HW_appendix}
    \begin{split}
        \langle \hat H_S' \hat H_W \rangle_{\hat \sigma} &= \Tr[ \hat H_S' \hat H_W \sum_{i,j} \int dE dE' \rho_{ij} \rho(E-\epsilon_i, E' - \epsilon_j) \dyad{\epsilon_i, E}{ \epsilon_j, E'}] \\
        &= \sum_{i,j} \rho_{ij} \int dE dE' \ E \ \rho(E-\epsilon_i, E' - \epsilon_j) \Tr[ \hat H_S'  \dyad{\epsilon_i, E}{ \epsilon_j, E'}] \\
        &= \sum_{i,j} \rho_{ij} \int dE \ E \ \rho(E-\epsilon_i, E - \epsilon_j) \Tr[ \hat H_S'  \dyad{\epsilon_i}{ \epsilon_j}] \\
        &\stackrel{(\ref{Wigner_iden})}{=} \sum_{i,j} \rho_{ij} \int dE \ [E+\frac{1}{2}(\epsilon_i + \epsilon_j)] \int dt e^{-i \omega_{ij} t} W(E,t) \Tr[ \hat H_S'  \dyad{\epsilon_i}{ \epsilon_j}] \\
        &= \sum_{i,j} \rho_{ij} \int dE \ [E+\frac{1}{2}(\epsilon_i + \epsilon_j)] \int dt W(E,t) \Tr[ \hat H_S'  e^{-i \hat H_S t} \dyad{\epsilon_i}{ \epsilon_j} e^{i \hat H_S t} ] \\
        &=  \int dE \int dt \ W(E,t) \big[E \Tr[ \hat H_S'  e^{-i \hat H_S t} \hat \rho_S e^{i \hat H_S t} ] \\
        &+ \frac{1}{2} \Tr[ \hat H_S' \{ e^{-i \hat H_S t} \hat \rho_S e^{i \hat H_S t}, \hat H_S \} ] \big] \\
        &= \int dE \int dt \ W(E,t) E \Tr[ \hat H_S'  e^{-i \hat H_S t} \hat \rho_S e^{i \hat H_S t} ] + \Tr[ \hat H_S' \{ \hat \sigma_S, \hat H_S \} ] \\
        &= \int dE \int dt \ W(E,t) E \Tr[ \hat H_S'  e^{-i \hat H_S t} \hat \rho_S e^{i \hat H_S t} ] + \frac{1}{2} \Tr[\{ \hat H_S', \hat H_S \} \hat \sigma_S ]
    \end{split}
\end{equation}
Finally, we are ready to show the first equality of \eqref{F_appendix}:
\begin{equation} \label{F_derived_appendix}
    \begin{split}
        F &= \Cov_{\hat \sigma} [\hat H_S', \hat H_S] - \Cov_{\hat \sigma} [\hat H_S', \hat H_W] \\
        &= \frac{1}{2} \langle \{ \hat H_S',  \hat H_S \} \rangle_{\hat \sigma} -  \langle \hat H_S' \rangle_{\hat \sigma}  \langle \hat H_S \rangle_{\hat \sigma} - \langle \hat H_S' \hat H_W \rangle_{\hat \sigma} + \langle \hat H_W \rangle_{\hat \sigma} \langle \hat H_S' \rangle_{\hat \sigma} \\
        &\stackrel{\eqref{HW_appendix}}{=} \frac{1}{2} \langle \{ \hat H_S',  \hat H_S \} \rangle_{\hat \sigma} - \langle \hat H_S' \hat H_W \rangle_{\hat \sigma} + \langle \hat H_W \rangle_{\hat \rho} \langle \hat H_S' \rangle_{\hat \sigma} \\
        &\stackrel{\eqref{VS_HS_HW_appendix}}{=} -\int dE \int dt \ W(E,t) E \Tr[ \hat H_S'  e^{-i \hat H_S t} \hat \rho_S e^{i \hat H_S t} ] + \langle \hat H_W \rangle_{\hat \rho} \langle \hat H_S' \rangle_{\hat \sigma} \\
        &= \int dE \int dt \ W(E,t) \left(\langle \hat H_W \rangle_{\hat \rho} - E \right) \Tr[ \hat H_S'  e^{-i \hat H_S t} \hat \rho_S e^{i \hat H_S t} ] \\
        &=  \int dE \int dt \ W(E,t) \left(E - \langle \hat H_W \rangle_{\hat \rho} \right) \Tr[ \hat W_S  e^{-i \hat H_S t} \hat \rho_S e^{i \hat H_S t} ].
    \end{split}
\end{equation}
To derive the second equality and finish the proof, we calculate a derivative:
\begin{equation}
\begin{split}
\hat \xi'_S (s) &= \frac{d}{d s} \left( \frac{\int dt \int dE \ e^{i E s} \ W(E,t) \ e^{-i \hat H_S t} \hat \rho_S e^{i \hat H_S t}}{\Tr[ e^{i \hat H_W s} \hat \rho_W]} \right)  \\
&= \frac{\int dt \int dE \ i E e^{i E s} \ W(E,t) \ e^{-i \hat H_S t} \hat \rho_S e^{i \hat H_S t} \Tr[ e^{i \hat H_W s} \hat \rho_W]}{\Tr[ e^{i \hat H_W s} \hat \rho_W]^2} \\
&- \frac{\int dt \int dE \ e^{i E s} \ W(E,t) \ e^{-i \hat H_S t} \hat \rho_S e^{i \hat H_S t} \Tr[i \hat H_W e^{i \hat H_W s} \hat \rho_W] }{\Tr[ e^{i \hat H_W s} \hat \rho_W]^2}.
\end{split}
\end{equation}
Now, putting $s=0$, we have
\begin{equation}
\begin{split}
\hat \xi'_S (0) &= \int dt \int dE \ i E \ W(E,t) \ e^{-i \hat H_S t} \hat \rho_S e^{i \hat H_S t}   \\ &- \int dt \int dE \ W(E,t) \ e^{-i \hat H_S t} \hat \rho_S e^{i \hat H_S t} \Tr[i \hat H_W \hat \rho_W] \\
&= i \int dt \int dE \ W(E,t) \ e^{-i \hat H_S t} \hat \rho_S e^{i \hat H_S t} \ \left( E  - \langle \hat H_W \rangle_{\hat \rho} \right),
\end{split}
\end{equation}
and it is seen that the expression $-i \Tr[\hat W_S \hat \xi'_S (0)]$ corresponds to Eq. \eqref{F_derived_appendix}.

\section{Proof of Proposition 2}

\subsection*{Part I ($\Delta \sigma_W^2$ for a qubit in terms of $g_{ij}$ coefficients)}

We write the control-marginal state in a diagonal form:
\begin{equation}
    \hat \sigma_S = p \dyad{\psi_0} + (1-p) \dyad{\psi_1},
\end{equation}
where without loss of generality we assume that $p\le \frac{1}{2}$. Next, the work operator can be express as:
\begin{equation}
    \hat W_S = - \frac{w}{1-2p} \dyad{\psi_0} +  \frac{w}{1-2p} \dyad{\psi_1} + \beta \dyad{\psi_0}{\psi_1} + \beta^* \dyad{\psi_1}{\psi_0},
\end{equation}
where we use the fact that it is a Hermitian and that $\Tr[\hat W_S] = 0$. Here, $w$ is real and $\beta$ complex. The square of the operator is equal to:
\begin{eqnarray}
    \hat W_S^2 = \left(\frac{w^2}{(1-2p)^2} + |\beta|^2 \right) \left(\dyad{\psi_0} +  \dyad{\psi_1} \right) - \frac{w \beta}{1-2p} \dyad{\psi_0}{\psi_1} + \frac{w \beta^*}{1-2p} \dyad{\psi_1}{\psi_0}.
\end{eqnarray}
Consequently, we have
\begin{eqnarray}
    \langle \hat W_S \rangle_{\hat \sigma_S} &=& - \frac{p w}{1-2p} + \frac{(1-p) w}{1-2p}  = w, \\ 
     \langle \hat W_S^2 \rangle_{\hat \sigma_S} &=& \frac{w^2}{(1-2p)^2} + |\beta|^2,  \\
    \Var_{\hat \sigma_S}[\hat W_S] &=& \frac{4 p (1-p)}{(1-2p)^2} w^2 + |\beta|^2,
\end{eqnarray}
where
\begin{equation} \label{beta_appendix}
\begin{split}
    |\beta|^2 &= |\bra{\psi_0} \hat W_S \ket{\psi_1}|^2 = |\bra{\psi_0} (\hat H_S - \hat H_S') \ket{\psi_1}|^2\\
    &= |\bra{\psi_0} \hat H_S \ket{\psi_1}|^2 +  |\bra{\psi_0} \hat H_S' \ket{\psi_1}|^2 \\ 
    &- \bra{\psi_0} \hat H_S \ket{\psi_1} \bra{\psi_1} \hat H_S' \ket{\psi_0} - \bra{\psi_0} \hat H_S' \ket{\psi_1} \bra{\psi_1} \hat H_S \ket{\psi_0}.
\end{split}
\end{equation}

Next, we calculate the $F$-term. For simplicity we assume that $\langle \hat H_W \rangle_{\hat \rho} = 0$, such that we have:
\begin{equation} \label{F_term_appendix}
\begin{split}
    F &= \Cov_{\hat \sigma} [\hat H_S', \hat H_S] - \Cov_{\hat \sigma} [\hat H_S', \hat  H_W] \\
    &= \frac{1}{2} \langle \{ \hat H_S',  \hat H_S \} \rangle_{\hat \sigma} -  \langle \hat H_S' \rangle_{\hat \sigma}  \langle \hat H_S \rangle_{\hat \sigma} - \langle \hat H_S' \hat H_W \rangle_{\hat \sigma} + \langle \hat H_W \rangle_{\hat \sigma} \langle \hat H_S' \rangle_{\hat \sigma} \\
    &= \frac{1}{2} \langle \{ \hat H_S',  \hat H_S \} \rangle_{\hat \sigma} -  \langle \hat H_S' \rangle_{\hat \sigma}  \langle \hat H_S \rangle_{\hat \sigma} - \langle \hat H_S' \hat H_W \rangle_{\hat \sigma} + (\langle \hat H_W \rangle_{\hat \rho} + \langle \hat H_S \rangle_{\hat \sigma}) \langle \hat H_S' \rangle_{\hat \sigma} \\
    &= \frac{1}{2} \langle \{ \hat H_S',  \hat H_S \} \rangle_{\hat \sigma} - \langle \hat H_S' \hat H_W \rangle_{\hat \sigma},
\end{split}
\end{equation}
and consequently
\begin{eqnarray} \label{delta_sigma_appendix}
        \Delta \sigma_W^2 = \Var_{\hat \sigma_S}[\hat W_S] + 2F = \Var_{\hat \sigma_S}[\hat W_S] + \langle \{ \hat H_S',  \hat H_S \} \rangle_{\hat \sigma} - 2\langle \hat H_S' \hat H_W \rangle_{\hat \sigma}.
\end{eqnarray}

We assume that $\hat H_S = \omega \dyad{\epsilon_1}_S$ (with some $\omega >0$), and we put
\begin{eqnarray}
    \bra{\psi_0} \hat H_S \ket{\psi_0} &= \varepsilon_0, \ \bra{\psi_1} \hat H_S \ket{\psi_1} = \varepsilon_1.
\end{eqnarray}
Then, we derive a bunch of the following equalities:
\begin{eqnarray}
        \bra{\psi_0} \hat H_S' \ket{\psi_0} &=& \bra{\psi_0} (\hat H_S - \hat W_S) \ket{\psi_0} =\varepsilon_0 + \frac{w}{1-2p},\\
        \bra{\psi_1} \hat H_S' \ket{\psi_1} &=& \bra{\psi_1} (\hat H_S - \hat W_S) \ket{\psi_1} =\varepsilon_1 - \frac{w}{1-2p}, \\
        |\bra{\psi_0} \hat H_S \ket{\psi_1}|^2 &=& \bra{\psi_0} \hat H_S \ket{\psi_1} \bra{\psi_1} \hat H_S \ket{\psi_0} = \omega^2 \bra{\psi_0} \dyad{\epsilon_1} \ket{\psi_1} \bra{\psi_1} \dyad{\epsilon_1} \ket{\psi_0} \nonumber \\
        &=& \bra{\psi_0} \hat H_S \ket{\psi_0} \bra{\psi_1} \hat H_S \ket{\psi_1} = \varepsilon_0 \varepsilon_1, \\
        |\bra{\psi_0} \hat H_S' \ket{\psi_1}|^2 &=&  \bra{\psi_0} \hat H_S' \ket{\psi_0} \bra{\psi_1} \hat H_S' \ket{\psi_1} = (\varepsilon_0 + \frac{w}{1-2p})(\varepsilon_1 - \frac{w}{1-2p}), \\
        \bra{\psi_0} \hat H_S \hat H_S' \ket{\psi_0} &=& \varepsilon_0 (\varepsilon_0 + \frac{w}{1-2p}) + \bra{\psi_0} \hat H_S \dyad{\psi_1} \hat H_S' \ket{\psi_0}, \\
        \bra{\psi_1} \hat H_S \hat H_S' \ket{\psi_1} &=& \varepsilon_1 (\varepsilon_1 - \frac{w}{1-2p}) + \bra{\psi_1} \hat H_S \dyad{\psi_0} \hat H_S' \ket{\psi_1}.
\end{eqnarray}
Using the above expression we rewritten Eq. \eqref{beta_appendix} in the form:
\begin{multline}
   \bra{\psi_0} \hat H_S \ket{\psi_1} \bra{\psi_1} \hat H_S' \ket{\psi_0} + \bra{\psi_0} \hat H_S' \ket{\psi_1} \bra{\psi_1} \hat H_S \ket{\psi_0} = \\ = \varepsilon_0 \varepsilon_1 + (\varepsilon_0 + \frac{w}{1-2p})(\varepsilon_1 - \frac{w}{1-2p}) - |\beta|^2,    
\end{multline}
and then, we calculate the first term of Eq. \eqref{F_term_appendix}:
\begin{equation}
\begin{split}
     \langle \{ \hat H_S', \hat H_S \} \rangle_{\hat \sigma_S} &= p \bra{\psi_0} \hat H_S \hat H_S' + \hat V_S^\dag \hat H \hat V_S \hat H_S \ket{\psi_0} + (1-p) \bra{\psi_1} \hat H_S \hat H_S' + \hat H_S' \hat H_S \ket{\psi_1} \\
     & = 2p \varepsilon_0 \left(\varepsilon_0 + \frac{w}{1-2p} \right) + 2(1-p) \varepsilon_1 \left(\varepsilon_1 - \frac{w}{1-2p} \right) \\
     &+ \bra{\psi_0} \hat H_S \ket{\psi_1} \bra{\psi_1} \hat H_S' \ket{\psi_0} + \bra{\psi_0} \hat H_S' \ket{\psi_1} \bra{\psi_1} \hat H  \ket{\psi_0}  \\
     &=  2p \varepsilon_0 \left(\varepsilon_0 + \frac{w}{1-2p} \right) + 2(1-p) \varepsilon_1 \left(\varepsilon_1 - \frac{w}{1-2p} \right) \\
     &+  \varepsilon_0\varepsilon_1 +  (\varepsilon_0 + \frac{w}{1-2p})(\varepsilon_1 - \frac{w}{1-2p})  - |\beta|^2.
\end{split}
\end{equation}
We are ready to calculate first two terms of Eq. \eqref{delta_sigma_appendix}, i.e.: 
\begin{equation}
\begin{split}
    &\Var_{\hat \sigma_S}[\hat W] + \langle \{ \hat H_S', \hat H_S \} \rangle_{\hat \sigma_S} = \\
    &= \frac{4 p (1-p)}{(1-2p)^2} w^2 + |\beta|^2 + 2 p \varepsilon_0 \left(\varepsilon_0 + \frac{w}{1-2p} \right) + 2 (1-p) \varepsilon_1 \left(\varepsilon_1 - \frac{w}{1-2p} \right) + \varepsilon_0\varepsilon_1  \\
    &+  \left(\varepsilon_0 + \frac{w}{1-2p} \right) \left(\varepsilon_1 - \frac{w}{1-2p} \right)  - |\beta|^2  \\
    &= \frac{4 p (1-p)}{(1-2p)^2} w^2 + 2 p \varepsilon_0 \left(\varepsilon_0 + \frac{w}{1-2p} \right) + 2 (1-p) \varepsilon_1 \left(\varepsilon_1 - \frac{w}{1-2p} \right) + \varepsilon_0\varepsilon_1 \\
    &+  \left(\varepsilon_0 + \frac{w}{1-2p} \right) \left(\varepsilon_1 - \frac{w}{1-2p} \right)  \\
    &= \frac{4 p - 4p^2 - 1}{(1-2p)^2} w^2 + 2 p \varepsilon_0 \left(\varepsilon_0 + \frac{w}{1-2p} \right) + 2 (1-p) \varepsilon_1 \left(\varepsilon_1 - \frac{w}{1-2p} \right) \\
    &+ 2\varepsilon_0\varepsilon_1 + \frac{w}{1-2p}(\varepsilon_1 - \varepsilon_0) \\
    &= - w^2 + 2 p_0 \varepsilon_0 \left(\varepsilon_0 + \frac{w}{p_1-p_0} \right) + 2 p_1 \varepsilon_1 \left(\varepsilon_1 - \frac{w}{p_1 - p_0} \right) + 2\varepsilon_0\varepsilon_1 + \frac{w}{p_1-p_0}(\varepsilon_1 - \varepsilon_0) \\
    &= - w^2 + \frac{1}{p_1-p_0}\left(\varepsilon_0 (2 p_0 - 1) - \varepsilon_1 (2 p_1 -1)\right) w + 2 (p_0 \varepsilon_0^2 + p_1 \varepsilon_1^2 + \varepsilon_0\varepsilon_1),
\end{split}
\end{equation}
where we put $p_0 = p$ and $p_1 = 1-p$ to express the result in a symmetric form. In order to calculate the final term $\langle \hat H_S' \hat H_W \rangle_{\hat \sigma}$ we define the matrix element:
\begin{eqnarray}
        \varrho_{ij}(E,E) = \bra{\psi_i, E} \hat \sigma \ket{\psi_j, E}
\end{eqnarray}
%we expand the control state in the basis of states $\ket{\psi_0}$ and $\ket{\psi_1}$, namely:
% \begin{equation}
%     \hat \sigma = \sum_{i,j} \int dx \int dy \ \varrho_{ij} (x,y) \dyad{\psi_i}{\psi_j} \otimes \dyad{x}{y},
% \end{equation}
% such that
% \begin{equation}
% \begin{split}
%     \hat \sigma_S = \Tr_W[ \hat \sigma] = \sum_{i,j} \int dx \varrho_{ij} (x,x) \dyad{\psi_i}{\psi_j}, \ \int dx \ \varrho_{ij} (x,x) = \delta_{ij} p_i.
% \end{split}
% \end{equation}
and the following coefficient:
\begin{equation}
    g_{ij} = \int dE  \ E \ \varrho_{ij}(E,E),
\end{equation}
such that 
\begin{equation}
\begin{split}
    \langle \hat H_S' \hat H_W \rangle_{\hat \sigma} &= \sum_{i,j} \int dx  \ x \ \rho_{ij}(x,x) \bra{\psi_j} \hat H' \ket{\psi_i} \\
    & = g_{00} (\varepsilon_0 + \frac{w}{1-2p}) + g_{11} (\varepsilon_1 - \frac{w}{1-2p}) + g_{01} \bra{\psi_1} \hat H' \ket{\psi_0} + g_{10} \bra{\psi_0} \hat H' \ket{\psi_1}.
\end{split}
\end{equation}
Let us put $g_{01} \bra{\psi_1} \hat H' \ket{\psi_0} = |g_{01} \bra{\psi_1} \hat H' \ket{\psi_0}| e^{i \phi}$ with some phase $\phi$. Hence, the sum of last two terms is equal to:
\begin{equation}
\begin{split}
         g_{01} \bra{\psi_1} \hat H' \ket{\psi_0} + g_{10} \bra{\psi_0} \hat H' \ket{\psi_1} &= 2 \text{Re}(g_{01} |\bra{\psi_1} \hat H' \ket{\psi_0}) = 2 |g_{01}| |\bra{\psi_1} \hat H' \ket{\psi_0}| \cos{\phi} \\
         &= 2 |g_{01}| \sqrt{(\varepsilon_0 + \frac{w}{1-2p})(\varepsilon_1 - \frac{w}{1-2p})} \cos{\phi},
\end{split}
\end{equation}
where 
\begin{eqnarray} \label{phi_phase_appendix}
\cos \phi = \frac{\text{Re}[g_{01} \bra{\psi_1} \hat H' \ket{\psi_0}]}{|g_{01} \bra{\psi_1} \hat H' \ket{\psi_0}|}.
\end{eqnarray}
Finally, we get
\begin{equation}
\begin{split}
    \langle \hat H_S' \hat H_W \rangle_{\hat \sigma} &= \frac{1}{1-2p} (g_{00} - g_{11}) w + g_{00} \varepsilon_0 + g_{11} \varepsilon_1 \\
    &+ 2 |g_{01}| \sqrt{(\varepsilon_0 + \frac{w}{1-2p})(\varepsilon_1 - \frac{w}{1-2p})} \cos{\phi}.
\end{split}
\end{equation}
and
\begin{equation} \label{variance_gain_appendix}
\begin{split}
    \Delta \sigma_W^2 &= - w^2 + \frac{1}{1-2p}\left[\varepsilon_0 (2 p_0 - 1) - 2g_{00} - \varepsilon_1 (2 p_1 -1) + 2g_{11} \right] w \\ &
    + 2 [\varepsilon_0 (p_0 \varepsilon_0 - g_{00}) + \varepsilon_1 (p_1 \varepsilon_1 - g_{11}) + \varepsilon_0\varepsilon_1 ] \\
    &+ 4 |g_{01}| \sqrt{(\varepsilon_0 + \frac{w}{1-2p})(\varepsilon_1 - \frac{w}{1-2p})} \cos{\phi}.
\end{split}
\end{equation}

\subsection*{Part II ($\hat \sigma_S$ and $\hat \rho_S$ relations)}
To complete the proof we want to express the coefficients $g_{ij}$ in terms of the parameters $\varepsilon_0, \varepsilon_1, p$ and integrals:
\begin{eqnarray}
        \eta &=& \frac{1}{\omega^2} \int dt \int dE \ E \ W(E,t) \ (e^{i \omega t} \frac{1}{\gamma} + e^{-i \omega t} \frac{1}{\gamma^*}), \\
        \xi  &=& \frac{1}{\omega^2} \int dt \int dE \ E \ W(E,t) \ (e^{i \omega t} \frac{\varepsilon_1}{\gamma } - e^{-i \omega t} \frac{\varepsilon_0}{\gamma^*} ), \\
        \gamma &=& \int dt \int dE \ W(E,t) \ e^{i \omega t}. \label{gamma_appendix}
\end{eqnarray}

Since we want to express the final result in terms of the eigenvectors of the control-marginal state (i.e., $\ket{\psi_0}, \ket{\psi_1}$), first, we need to derive its relation to the initial state $\hat \rho_S$. Let us represent the initial density matrix $\hat \rho_S$ in the energetic basis $\ket{\epsilon_0} \equiv \ket 0, \ket{\epsilon_1} \equiv \ket 1 $ by real value $p_0$ (diagonal) and complex parameter $\alpha$ (off-diagonal). Then, we consider a map:
\begin{equation}
\hat \rho_S = 
    \begin{pmatrix}
    p_0 & \alpha \\
    \alpha^*  & 1-p_0
    \end{pmatrix}
    \to
    \begin{pmatrix}
    p_0 & \gamma \alpha \\
    \gamma^* \alpha^*  & 1-p_0
    \end{pmatrix} = \hat \sigma_S
\end{equation}
where $\gamma$ is given by Eq. \eqref{gamma_appendix}. By representing the vectors $\ket{\psi_0}, \ket{\psi_1}$ in the energetic basis and equations:
\begin{eqnarray} \label{eigenvalue_problem_appendix}
        \hat \sigma_S \ket{\psi_0} = p, \ \ \hat \sigma_S \ket{\psi_1} = 1-p. 
\end{eqnarray}
we are able to find relations between parameters $p, \varepsilon_0, \varepsilon_1, p_0$ and $\alpha$. 

First, from a definition of $\varepsilon_i$, we have:
\begin{equation}
    \varepsilon_i = \bra{\psi_i} \hat H_S \ket{\psi_i} = \omega |\braket{\psi_i}{1}|^2 
\end{equation}
and since 
\begin{equation}
    \varepsilon_0 + \varepsilon_1 = \Tr[\hat H_S (\dyad{\psi_0} + \dyad{\psi_1})] = \Tr[\hat H_S] = \omega,
\end{equation}
we can write the following:
\begin{equation}
  |\braket{\psi_i}{1}| = \sqrt{\frac{\varepsilon_i}{\omega}}, \ \ |\braket{\psi_i}{0}| =\sqrt{1-\frac{\varepsilon_i}{\omega}}.
\end{equation}
Then, by introducing the relative phases $\phi_0$ and $\phi_1$ we have
\begin{eqnarray}
    \ket{\psi_0} &=& \braket{0}{\psi_0} \ket{0} + \braket{1}{\psi_0} \ket{1} = \sqrt{\frac{\varepsilon_1}{\omega}} \ket{0} + e^{i \phi_0} \sqrt{\frac{\varepsilon_0}{\omega}} \ket{1}, \\
    \ket{\psi_1} &=& \braket{0}{\psi_1} \ket{0} + \braket{1}{\psi_1} \ket{1} =  \sqrt{\frac{\varepsilon_0}{\omega}} \ket{0} + e^{i \phi_1} \sqrt{\frac{\varepsilon_1}{\omega}} \ket{1}.
\end{eqnarray}
In accordance, the first equation of \eqref{eigenvalue_problem_appendix}, in the matrix form, is given by:
\begin{equation}
    \begin{pmatrix}
    p_0 & \gamma \alpha \\
    \gamma^* \alpha^*  & 1-p_0
    \end{pmatrix}
    \begin{pmatrix}
     \sqrt{\varepsilon_1}  \\
     e^{i \phi_0} \sqrt{\varepsilon_0}
    \end{pmatrix}
    = p
    \begin{pmatrix}
      \sqrt{\varepsilon_1}  \\
    e^{i \phi_0} \sqrt{\varepsilon_0}
    \end{pmatrix}.
\end{equation}
This results in condition:
\begin{eqnarray} 
    e^{i \phi_0} = \frac{1}{\gamma \alpha} \sqrt{\frac{\varepsilon_1}{\varepsilon_0}}(p-p_0).
\end{eqnarray}
Furthermore, the energy $\Tr[\hat H_S \hat \sigma_S]$, expressed in two different basis, leads us to the equality: $\omega (1-p_0) = p \varepsilon_0 + (1-p) \varepsilon_1$, such that
\begin{eqnarray}
    p_0 = 1-p - (1-2p) \frac{\varepsilon_1}{\omega}
\end{eqnarray}
and finally we have
\begin{eqnarray} \label{phase_appendix}
    e^{i \phi_0} = - \frac{\sqrt{\varepsilon_0 \varepsilon_1}(1-2p)}{\omega \gamma \alpha}.
\end{eqnarray}
In the same way, using the second equation \eqref{eigenvalue_problem_appendix}, we obtain $e^{i \phi_1} = - e^{i \phi_0}$. Finally, we arrive with formulas:
\begin{equation} \label{overlaps_appendix}
    \begin{split}
        &\braket{0}{\psi_0} = \sqrt{\frac{\varepsilon_1}{\omega}}, \ \ \braket{1}{\psi_0} = - \frac{(1-2p) \sqrt{\varepsilon_0 \varepsilon_1}}{\omega \gamma \alpha} \sqrt{\frac{\varepsilon_0}{\omega}}, \\
        &\braket{0}{\psi_1} = \sqrt{\frac{\varepsilon_0}{\omega}}, \ \ \braket{1}{\psi_1} = \frac{(1-2p) \sqrt{\varepsilon_0\varepsilon_1}}{\omega \gamma \alpha } \sqrt{\frac{\varepsilon_1}{\omega}},
    \end{split}
\end{equation}
and since $|e^{i \phi_{0,1}}| = 1$, we have also
\begin{equation} \label{absolute_value_appendix}
    \left|\frac{(1-2p) \sqrt{\varepsilon_0 \varepsilon_1}}{\omega \gamma \alpha} \right| = 1.
\end{equation}

\subsection*{Part III (expressions for $g_{ij}$ coefficients)}
Now, we are ready to finish the proof. First, we put the expression \eqref{control_state_appendix} to the definition of the element $\varrho_{ij}(E,E)$: 
\begin{equation}
\begin{split}
    \varrho_{ij}(E,E) &= \bra{\psi_i, E} \hat \sigma \ket{\psi_j, E} =  \sum_{k,l} \rho_{kl} \rho(E-\epsilon_k, E-\epsilon_l) \braket{\psi_i}{\epsilon_k} \braket{\epsilon_l}{\psi_j} 
\end{split}
\end{equation}
which, as a consequence, brings us the formula:
\begin{equation} \label{g_ij_appendix}
\begin{split}
    g_{ij} &= \int dE  \ E \ \varrho_{ij}(E,E) =  \sum_{k,l} \rho_{kl}  \braket{\psi_i}{\epsilon_k} \braket{\epsilon_l}{\psi_j} \ \int dE  \ E \ \rho(E-\epsilon_k, E-\epsilon_l) \\
    &= \sum_{k,l} \rho_{kl}  \braket{\psi_i}{\epsilon_k} \braket{\epsilon_l}{\psi_j} \ \int dE \ [E + \frac{1}{2} (\epsilon_k + \epsilon_l)] \int dt \ e^{-i\omega_{kl} t} \ W(E,t) \\
    &= \int dE \int dt \ E \ W(E,t) \bra{\psi_i} e^{-i \hat H_S t} \hat \rho_S e^{i \hat H_S t} \ket{\psi_j} \\
    &+ \frac{1}{2} \int dE \int dt \ W(E,t) \bra{\psi_i} \{\hat H_S, e^{-i \hat H_S t} \hat \rho_S e^{i \hat H_S t}  \} \ket{\psi_j} \\
    &= \int dE \int dt \ E \ W(E,t) \bra{\psi_i} e^{-i \hat H_S t} \hat \rho_S e^{i \hat H_S t} \ket{\psi_j} + \frac{1}{2}  \bra{\psi_i} \{\hat H_S, \hat \sigma_S \} \ket{\psi_j}.
\end{split}
\end{equation}
To simplify notation, we put 
\begin{equation}
\begin{split}
        g_{ij} &= k_{ij} + h_{ij}, \\
        k_{ij} &= \frac{1}{2}  \bra{\psi_i} \{\hat H_S, \hat \sigma_S \} \ket{\psi_j}, \\
        h_{ij} &= \int dE \int dt \ E \ W(E,t) \bra{\psi_i} e^{-i \hat H_S t} \hat \rho_S e^{i \hat H_S t} \ket{\psi_j}.
\end{split}
\end{equation}
Putting the expression $\hat \sigma_S = p_0 \dyad{\psi_0} + p_1 \dyad{\psi_1}$ (where $p_0 = p$ and $p_1 = 1-p$), we have
\begin{equation}
\begin{split}
        k_{ij} &= \frac{1}{2} \bra{\psi_i} \left(\sum_{k=0,1} p_k \dyad{\psi_k} \hat H_S + p_k \hat H_S \dyad{\psi_k} \right) \ket{\psi_j} \\
        &= \frac{1}{2} \sum_{k=0,1} p_k \left(\delta_{i,k} \bra{\psi_k}\hat H_S \ket{\psi_j} + \delta_{j,k} \bra{\psi_i}\hat H_S \ket{\psi_k}  \right)  = \frac{1}{2} (p_i + p_j) \bra{\psi_i}\hat H_S \ket{\psi_j},
\end{split}
\end{equation}
such that
\begin{equation}
    \begin{split}
        k_{ii} = p_i \varepsilon_i, \ \ k_{01} &= \frac{1}{2} \bra{\psi_0} \hat H_S \ket{\psi_1} = \frac{\omega}{2} \braket{\psi_0}{1} \braket{1}{\psi_1} \\
        &=  - \frac{1}{2} \sqrt{\varepsilon_0 \varepsilon_1}  \left|\frac{(1-2p) \sqrt{\varepsilon_0 \varepsilon_1}}{\omega \gamma \alpha} \right|^2 = - \frac{1}{2} \sqrt{\varepsilon_0 \varepsilon_1},
    \end{split}
\end{equation}
where we used Eq. \eqref{overlaps_appendix} and \eqref{absolute_value_appendix}.

% \begin{equation}
%     g(t) = \int dE \ W(E, t)
% \end{equation}
% \begin{equation}
%     s(t) = \int dE \ E \ W(E, t), \ \ \int dt \ s(t) = \int dE \ E \ \int dt \ W(E,t) = \int dE \ E \ f(E) =  0
% \end{equation}
% \begin{equation}
%     \hat \rho_S = p_0 \dyad{0} + (1-p_0) \dyad{1} + \alpha \dyad{0}{1} + \alpha^* \dyad{1}{0}
% \end{equation}
Next, we consider the second contribution to $g_{ij}$, namely 
\begin{equation}
\begin{split}
    h_{ij} &= \int dE \int dt \ E \ W(E,t) \bra{\psi_i} e^{-i \hat H_S t} \hat \rho_S e^{i \hat H_S t} \ket{\psi_j} \\
    &= \int dt \int dE \ E \ W(E, t) (\alpha e^{i \omega t} \braket{\psi_i}{0} \braket{1}{\psi_j} + \alpha^* e^{-i \omega t} \braket{\psi_i}{1} \braket{0}{\psi_j}).
\end{split}
\end{equation}
Substituting results from Eq. \eqref{overlaps_appendix}, we have 
\begin{equation}
\begin{split}
    h_{ii} &= \int dt \int dE \ E \ W(E, t) (\alpha e^{i \omega t} \braket{\psi_i}{0} \braket{1}{\psi_i} + \alpha^* e^{-i \omega t} \braket{\psi_i}{1} \braket{0}{\psi_i}  ) \\
    &= \frac{\sqrt{\varepsilon_0 \varepsilon_1}}{\omega} \int dt \int dE \ E \ W(E, t) (\alpha e^{i \omega t} e^{i \phi_i} + \alpha^* e^{-i \omega t} e^{-i \phi_i}), \\
    h_{01} &= \frac{1}{\omega} \int dt \int dE \ E \ W(E, t) (\alpha e^{i \omega t} e^{i\phi_1} \varepsilon_1 + \alpha^* e^{-i \omega t} e^{-i\phi_0} \varepsilon_0).
\end{split}    
\end{equation}
Notice that since $e^{i \phi_0} = -e^{i \phi_1}$, then $h_{00} = - h_{11}$. Incorporating this and Eq. \eqref{phase_appendix}, we get the following formulas:
\begin{equation}
\begin{split}
        h_{11} &= - h_{00} = \frac{(1-2p) \varepsilon_0 \varepsilon_1}{\omega^2} \int dt \int dE \ E \ W(E, t) (e^{i \omega t} \frac{1}{\gamma} + e^{-i \omega t} \frac{1}{\gamma^*}) = (1-2p) \varepsilon_0 \varepsilon_1  \eta, \\
       h_{01} &= \frac{(1-2p)\sqrt{\varepsilon_0\varepsilon_1}}{\omega^2} \int dt \ s(t) (e^{i \omega t} \frac{\varepsilon_1}{\gamma } - e^{-i \omega t} \frac{\varepsilon_0}{\gamma^*} ) = (1-2p)\sqrt{\varepsilon_0\varepsilon_1} \xi.
\end{split}
\end{equation}

Then, we put the derived expressions for $g_{ij}$ to the formula \eqref{variance_gain_appendix}, namely
\begin{equation}
\begin{split}
    \Delta \sigma_W^2 &= - w^2 + \frac{1}{1-2p}\left[\varepsilon_0 (2 p_0 - 1) - 2(k_{00}+h_{00}) - \varepsilon_1 (2 p_1 -1) + 2(k_{11}+h_{11}) \right] w  \\
    &+2 [\varepsilon_0 (p_0 \varepsilon_0 - k_{00}-h_{00}) + \varepsilon_1 (p_1 \varepsilon_1 - k_{11}-h_{11}) + \varepsilon_0\varepsilon_1 ] \\
    &+ 4 |k_{01}+h_{01}| \sqrt{(\varepsilon_0 + \frac{w}{1-2p})(\varepsilon_1 - \frac{w}{1-2p})} \cos{\phi} \\
    &= - w^2 + \frac{1}{1-2p}\left[ \varepsilon_1 - \varepsilon_0  - 2h_{00} + 2h_{11} \right] w -2 [\varepsilon_0 h_{00} + \varepsilon_1 h_{11} - \varepsilon_0\varepsilon_1 ] \\
    &+ 4 |k_{01}+h_{01}| \sqrt{(\varepsilon_0 + \frac{w}{1-2p})(\varepsilon_1 - \frac{w}{1-2p})} \cos{\phi} \\
     &= - w^2 + \frac{1}{1-2p}\left[ \varepsilon_1 - \varepsilon_0  + 4 (1-2p) \varepsilon_0 \varepsilon_1 \eta \right] w -2 \varepsilon_0 \varepsilon_1 [(\varepsilon_1 - \varepsilon_0)(1-2p)\eta - 1] \\
     &+ 2|1 - 2(1-2p) \xi| \sqrt{\varepsilon_0 \varepsilon_1(\varepsilon_0 + \frac{w}{1-2p})(\varepsilon_1 - \frac{w}{1-2p})} \cos{\phi} \\
     &= - w^2 + \left(\frac{\varepsilon_1 - \varepsilon_0}{1-2p} + 4 \varepsilon_0 \varepsilon_1 \eta \right) w + 2 \varepsilon_0 \varepsilon_1 [1 - (\varepsilon_1 - \varepsilon_0)(1-2p)\eta] \\
     &+ 2|1 - 2(1-2p) \xi| \sqrt{\varepsilon_0 \varepsilon_1(\varepsilon_0 + \frac{w}{1-2p})(\varepsilon_1 - \frac{w}{1-2p})} \cos{\phi}. \\
\end{split}
\end{equation}
Finally, since $\cos \phi \in [-1,1]$ we proved that
\begin{equation} \label{variance_set_appendix}
\begin{split}
     \Delta \sigma_W^2 \in& - w^2 + \left(\frac{\varepsilon_1 - \varepsilon_0}{1-2p} + 4 \varepsilon_0 \varepsilon_1 \eta \right) w + 2 \varepsilon_0 \varepsilon_1 [1 - (\varepsilon_1 - \varepsilon_0)(1-2p)\eta] \\
 &\pm 2R \sqrt{\varepsilon_0 \varepsilon_1(\varepsilon_0 + \frac{w}{1-2p})(\varepsilon_1 - \frac{w}{1-2p})} 
\end{split}
\end{equation}
where $R = |1 - 2(1-2p) \xi|$. 
%that for an arbitrary initial state $\hat \rho_S \otimes \hat \rho_W$ and arbitrary unitary $\hat V_S$, the corresponding change of a weight's energy and variance, i.e., $(\Delta E_W, \Delta \sigma_W^2)$, lies within a set proved by Proposition \ref{qubit_proposition}. 

\subsection*{Part IV (range of $\Delta E_W$ and $\Delta \sigma_W^2$)}
Let us now prove that for a fixed state $\hat \sigma_S$, a possible values of $w = \Tr[\hat W_S \hat \sigma_S]$ belongs to a set $[-(1-2p) \varepsilon_0, (1-2p) \varepsilon_1]$. Notice that the maximal value of $w$ corresponds to ergotropy of a state $\hat \sigma_S$. To calculate the ergotropy, commonly the density matrix is diagonalized, i.e., one searches for a minimal energy state with fixed spectrum, which is called the passive state. However, since here we work in a basis $\{ \ket{\psi_i} \}$, such that a state $\hat \sigma_S$ is already diagonal, it is easier to diagonalize the final Hamiltonian $\hat H_S' = \hat V_S^\dag \hat H_S \hat V_S$ in this basis. Moreover, since we take the Hamiltonian in a form $\hat H_S = \omega \dyad{1}_S$, we already know its spectrum (i.e., $0$ and $\omega$ are the eigenvalues). Consequently, since we assume that $p\le \frac{1}{2}$, the minimal (maximal) value $w$ is extracted if $\hat H_S' = \omega \dyad{\psi_1}_S$ ($\hat H_S' = \omega \dyad{\psi_0}_S$), namely
\begin{eqnarray}
w_{max} &=& \Tr[\hat W_S \hat \sigma_S] = \Tr[\hat H_S \hat \sigma_S] - \Tr[\hat H_S' \hat \sigma_S] \nonumber \\
&=& p \varepsilon_0 + (1-p) \varepsilon_1 - p (\varepsilon_0 + \varepsilon_1) = (1-2p) \varepsilon_1, \\
w_{min} &=& \Tr[\hat W_S \hat \sigma_S] = \Tr[\hat H_S \hat \sigma_S] - \Tr[\hat H_S' \hat \sigma_S] \nonumber \\
&=& p \varepsilon_0 + (1-p) \varepsilon_1 - (1-p) (\varepsilon_0 + \varepsilon_1) = -(1-2p) \varepsilon_0,
\end{eqnarray}
where we used $\omega = \varepsilon_0 + \varepsilon_1$. Furthermore, the unitary operator $\hat V_S$ can be parameterized by a continuous variables (i.e., angles and phases), hence, there always exist a unitary operator $\hat V_S$, such that $\Tr[\hat W_S \hat \sigma_S] = w$ for arbitrary $w \in [-(1-2p) \varepsilon_0, (1-2p) \varepsilon_1]$. 

Next, we prove that for a fixed $w$, there is a unitary $\hat V_S$, such that $\cos \phi$ takes an arbitrary value from a set $[-1,1]$, i.e., the gain of a variance $\Delta \sigma_W^2$ covers the whole set \eqref{variance_set_appendix}.   Let us consider a unitary operator $\hat V_S(t) = \hat V_S e^{i \dyad{\psi_0} t}$ for arbitrary real $t$, where $\hat V_S$ is also a unitary operator. Then, we observe that for arbitrary $t_1, t_2$, we have 
\begin{eqnarray}
\Tr[\hat W_S(t_1) \hat \sigma_S] = \Tr[\hat W_S(t_2) \hat \sigma_S] \equiv w,
\end{eqnarray}
where $\hat W_S(t) = \hat H_S - \hat V_S^\dag(t) \hat H_S \hat V_S(t)$, since  $[e^{i \dyad{\psi_0} t}, \hat \sigma_S] = 0$. Next, we have also
\begin{eqnarray}
\bra{\psi_1} \hat H' \ket{\psi_0} = \bra{\psi_1} e^{-i \dyad{\psi_0} t} \hat V_S^\dag \hat H_S \hat V_S e^{i \dyad{\psi_0} t} \ket{\psi_0} = e^{i t} \bra{\psi_1} \hat V_S^\dag \hat H_S \hat V_S \ket{\psi_0}.
\end{eqnarray}
Putting it into the definition $\eqref{phi_phase_appendix}$, we get $\phi = t + \delta$, where $\delta$ is a phase of $g_{01} \bra{\psi_1} \hat V_S^\dag \hat H_S \hat V_S \ket{\psi_0}$, and since $t$ is arbitrary, it completes the proof.  

% Finally, let us back to the definition of a phase $\phi$, i.e.:  
% \begin{eqnarray}
% \cos \phi = \frac{\text{Re}[g_{01} \bra{\psi_1} \hat H' \ket{\psi_0}]}{|g_{01} \bra{\psi_1} \hat H' \ket{\psi_0}|}.
% \end{eqnarray}
% Then, for a fixed initial state $\hat \rho_S \otimes \hat \rho_W$, the coefficient $g_{01}$ and vectors $\ket{\psi_0}, \ket{\psi_1}$ are fixed.  

\section{Coherent vs Incoherent work extraction: Gaussian wave packet and Qubit} \label{coherent_incoherent_appendix}
For the arbitrary Gaussian state of the weight in the form \eqref{gaussian_states}, the characteristic function \eqref{gamma_function} is given by:
\begin{eqnarray} \label{characteristic_exponential_appendix}
\gamma (\omega) = \int dt \ g(t) \ e^{i \omega t} = e^{-\frac{1}{2} \sigma_t^2 \omega^2 - i \omega \langle t \rangle}.
\end{eqnarray}
where $\sigma_t^2$ is a variance and $\langle t \rangle$ is an expected value of the distribution $g(t)$. Next, we consider a qubit in a state:
\begin{eqnarray}
\hat \rho_S = \frac{1}{2} (\mathbb{1} + \vec x \cdot \vec{\hat \sigma}).
\end{eqnarray}
Its total ergotropy is equal to (see e.g., \cite{Lobejko2020}):
\begin{eqnarray}
R = \frac{\omega}{2} \left(|\vec x| - z \right).
\end{eqnarray}
From this, we also immediately obtain the incoherent part putting $x=y=0$, namely
\begin{equation}
    R_I = \frac{\omega}{2} \left(|z| - z \right),
\end{equation}
and then the coherent contribution by a difference, i.e.:
\begin{eqnarray}
R_C = R - R_I = \frac{\omega}{2} \left(|\vec x| - |z| \right).
\end{eqnarray}
Next, the control-marginal state $\hat \sigma_S$ represented in the Pauli matrix basis has a vector:
\begin{eqnarray}
\vec x_C = \left(|\gamma(\omega)| (x \cos(\omega \langle t \rangle) - y \sin(\omega \langle t \rangle)), |\gamma(\omega)| (x \cos(\omega \langle t \rangle) + y \sin(\omega \langle t \rangle)), z \right)
\end{eqnarray}
such that 
\begin{eqnarray}
|\vec x_C| = \sqrt{|\gamma(\omega)|^2(x^2+y^2) + z^2}.
\end{eqnarray}
Consequently the coherent ergotropy of the control-marginal state is equal to:
\begin{equation}
     R_C = \frac{\omega}{2} \left(\sqrt{|\gamma(\omega)|^2 \alpha^2 + z^2} - |z|\right)
\end{equation}
where we put $x^2+y^2 = \alpha^2$. From this, we have 
\begin{equation}
  |\gamma(\omega)|^2 =\frac{4 R_C(R_C + |z| \omega)}{\omega^2 \alpha^2}, 
\end{equation}
and then putting Eq. \eqref{characteristic_exponential_appendix} we get
% \begin{equation}
%  e^{-\sigma_t^2 \omega^2} = \frac{4 R_C(R_C + |z| \omega)}{\omega^2 \alpha^2}
% \end{equation}
\begin{equation}
 \sigma_t = \frac{1}{\omega} \sqrt{\log[\frac{\omega^2 \alpha^2}{4 R_C(R_C + |z| \omega)}]}.
\end{equation}
Finally, we apply HUR \eqref{HUR} for a final state of the weight, where $\sigma_t = \sigma_t^{(i)} = \sigma_t^{(f)}$, i.e.,
\begin{eqnarray}
\sigma_E^{(f)} \sigma_t \ge \frac{1}{2},
\end{eqnarray}
which results in inequality:
\begin{eqnarray}
\sigma_E^{(f)} \ge \frac{\omega}{4 \sqrt{\log[\frac{\omega^2 \alpha^2}{4 R_C(R_C + |z| \omega)}]}}.
\end{eqnarray}

Now, let us derive a lower bound for $\sigma_E^{(f)}$ if the process is incoherent. In this case, we have
\begin{eqnarray}
\sigma_E^{(f)} \ge \sqrt{\Var_{\hat \rho_S} [\hat W_S]}.
\end{eqnarray}
Next, we assume that total incoherent ergotropy is extracted, i.e., $\langle \hat W_S \rangle_{\hat \rho_S} = R_I$. This however is done only through the incoherent operation that swap the ground level with excited level (i.e., $\hat V_S = \hat \sigma_x$). Then,
\begin{eqnarray}
\hat W_S = \hat H_S - \hat V_S^\dag \hat H_S \hat V_S = \hat H_S - \hat \sigma_x \hat H_S \hat \sigma_x = \omega \left(\dyad{1}_S - \dyad{0}_S \right) = - \omega \hat \sigma_z
\end{eqnarray}
such that $\hat W_S^2 = \omega^2$. Finally, we have
\begin{eqnarray}
\sigma_E^{(f)} \ge \sqrt{\Var_{\hat \rho_S} [\hat W_S]} = \sqrt{\langle \hat W_S^2 \rangle_{\hat \rho_S} - \langle \hat W_S \rangle_{\hat \rho_S}^2} = \sqrt{\omega^2 - R_I^2}.
\end{eqnarray}

% Now, we apply the result of Theorem \ref{fluctuation_decoherence_theorem}. First, we have 
% \begin{equation}
%     \frac{1}{\pi} \max_{\omega} \big[ \omega |\gamma(\omega)| \big] = \frac{e^{-\frac{1}{2}}}{\pi \sigma_t} = \frac{2 \sigma_E^{(i)}}{\sqrt{e} \pi}
% \end{equation}

\end{document}